%% file: 01_main_v2.tex
\documentclass[prx, aps,reprint,superscriptaddress, longbibliography]{revtex4-2}

\usepackage{braket, bm, enumerate}
\usepackage{amsmath,amssymb,amsthm,mathrsfs,amsfonts,dsfont,mathtools,mathrsfs}
\usepackage{physics}
\usepackage{graphicx}
\usepackage{thmtools}

\usepackage[colorlinks=true,citecolor=blue,linkcolor=blue,urlcolor=blue]{hyperref}
\usepackage[capitalise]{cleveref}
\usepackage{multirow}

\usepackage{algpseudocode}
\usepackage[ruled,vlined,linesnumbered]{algorithm2e}
\algnewcommand{\Input}{\item[\textbf{Input:}]}
\algnewcommand{\Output}{\item[\textbf{Output:}]}

\usepackage{xcolor}
\usepackage{comment}
\usepackage{url}

\renewcommand{\tr}{\operatorname{Tr}}

\newcommand{\ddist}{d_{\diamond}}
\newcommand{\pphys}{p_{\rm phys}}
\newcommand{\barthetatarget}{\bar{\theta}_{\rm target}}

\newcommand{\mathC}{\mathcal{C}}
\newcommand{\mathI}{\mathcal{I}}
\newcommand{\mathM}{\mathcal{M}}
\newcommand{\mathR}{\mathcal{R}}

\newcommand{\mathU}{\mathcal{U}}

\newcommand{\bartheta}{\bar{\theta}}

\newcommand{\barP}{\bar{P}}
\newcommand{\supp}{{\rm supp}}

\newtheorem{theorem}{Theorem}
\newtheorem{definition}{Definition}
\newtheorem{lemma}{Lemma}

\newtheorem{proposition}{Proposition}

\newcommand{\black}[1]{\textcolor{black}{#1}}

\begin{document}

\title{Theory and Architecture of Syndrome-Resolved Logical Gates}

\author{Nobuyuki Yoshioka}
\email{ny.nobuyoshioka@gmail.com}
\affiliation{International Center for Elementary Particle Physics, The University of Tokyo, 7-3-1 Hongo, Bunkyo-ku, Tokyo 113-0033, Japan}

\author{Alireza Seif}
\affiliation{IBM Quantum, IBM T.J. Watson Research Center, Yorktown Heights, NY, USA}

\author{Peter Groszkowski}
\affiliation{National Center for Computational Sciences, Oak Ridge National Laboratory, Oak Ridge, Tennessee 37831, USA}

\author{Andrew Cross}
\affiliation{IBM Quantum, IBM T.J. Watson Research Center, Yorktown Heights, NY, USA}

\author{Ali Javadi-Abhari}
\affiliation{IBM Quantum, IBM T.J. Watson Research Center, Yorktown Heights, NY, USA}

\begin{abstract}
We introduce a general framework for weak transversal gates---syndrome-resolved implementation of logical unitaries realized by local physical unitaries---and demonstrate its wide impact in three applications: partial fault-tolerant architecture, short-depth state preparation, and magic state distillation.
Our theoretical contribution is to prove a sufficient condition for a general Calderbank-Shor-Steane code to admit weak transversal gates, and to present an efficient algorithm to determine the physical multiqubit Pauli rotations.
To demonstrate the practicality of weak transversal gates in the near future,
we propose a partially fault-tolerant Clifford+$\phi$ architecture that implements in-place Pauli rotations via a repeat-until-success strategy. Numerical simulations indicate that a rotation of $0.001$ attains a logical error of $9.6\times10^{-5}$ on a surface code with $d=7$ at a physical error rate of $\pphys=10^{-4}$, while avoiding the spacetime overheads of magic state factories, small angle synthesis, and routing. Resource estimates on surface and [[144,12,12]] bivariate bicycle codes for a Trotter-like circuit with $N=108$ logical qubits indicate runtime reductions of factors of 778 and 45, respectively, relative to the conventional Clifford+T execution, due to the natural parallelism of rotation gates.
\black{Looking further ahead, we consider state-preparation tasks allowing postselection. We show that weak transversal gadgets prepare high-level equatorial magic states with a substantially improved cost-accuracy tradeoff, achieving up to a $10^3$-fold reduction in spacetime cost for a magic state at the 11th level of the Clifford hierarchy on a surface code of $d=11$.
Furthermore, we integrate weak transversal gates into magic state distillation pipelines, reducing the logical error rate of level-2 distillation of the $\sqrt{T}$ state by up to a factor of 94 with a small increase in spacetime cost.
}
This work opens a novel paradigm for non-Clifford operations beyond conventional transversal and factory-based approaches.
\end{abstract}

\maketitle
\onecolumngrid

\section{Introduction} \label{sec:intro}

Transversal gates, the logical gates that can be implemented with only local operations on physical qubits, play a central role in fault-tolerant quantum computation (FTQC).
Beyond their direct use as fault-tolerant gate operations, they are also fundamental in protocols such as distillation of quantum resources~\cite{BravyiKitaev2004, Knill2004PostselectedQC}.
However, systematic methods to identify transversal gates in the conventional sense remain limited: a naive approach requires exponential-time computation~\cite{babai2011code, webster2023transversal}, and consequently, most known constructions have been derived on a case-by-case basis~\cite{BravyiHaah2012, Bombin2015GaugeColorCodes, paetznick2013universal, vasmer2019three}.
Considering that practical architectures of FTQC require resilience to qubit dropout~\cite{nagayama2017surface, auger2017fault, strikis2023quantum,  debroy2024luci, leroux2024snakes} or the use of dynamically changing error correcting codes (QECCs)~\cite{hastings2021dynamically, eickbusch2024demonstrating}, a more flexible and efficient framework for determining the feasibility of gate implementations is desired. In particular, in the earliest realizations of FTQC, combinations of fault-tolerant and non-fault-tolerant gates are envisioned to outperform full-fledged FTQC~\cite{choi2023fault, gavriel2023transversal, akahoshi2024partially, toshio2025practical}, and hence it is essential to develop methodologies that reduce the burden on quantum computers.

\begin{figure*}[t]
    \centering
\includegraphics[width=1.\linewidth]    {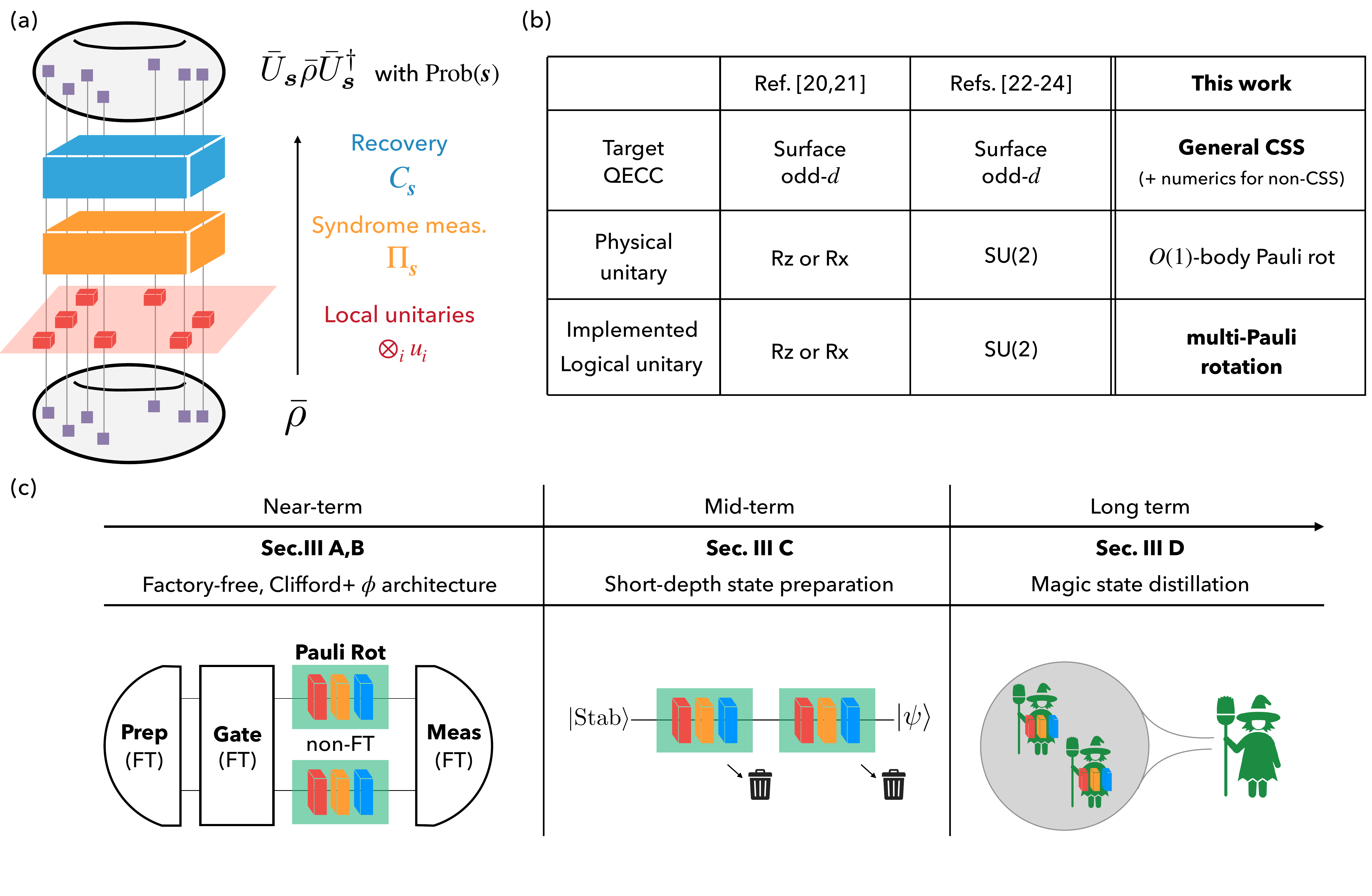}
    \caption{
(a) Concept of a weak transversal gate. If a set of given local unitaries followed by syndrome measurement with result $\bm s$ and correction operation $C_{\bm s}$ results in a logical unitary operation $\bar{U}_{\bm s}$ that is {\it independent of the input state}, then we call such an operation weakly transversal.  (b) Comparison with existing works that study weak transversal gates~\cite{bravyi2018correcting, suzuki2017efficient, cheng2024emergent, PhysRevLett.131.060603, behrends2024statistical}. Our main contribution in Sec.~\ref{sec:theory} is to establish a general condition for CSS codes to admit weak transversal implementation of multi-body Pauli rotation gates.
\black{(c) Applications considered in this work: partially fault-tolerant architecture (Secs.~\ref{subsec:inplace-rot},~\ref{subsec:resource-est}), short-depth state preparation (Sec.~\ref{subsec:stateprep}), and magic state distillation (Sec.~\ref{subsec:msd}), which we expect to be relevant in the near-, mid-, and long-term FTQC, respectively.}
     }
    \label{fig:fig1}
\end{figure*}

An emerging question is whether there exist transversal yet non-fault-tolerant gates.
Indeed, it has been found in the context of error correcting memory~\cite{bravyi2018correcting,suzuki2017efficient, PhysRevLett.131.060603,behrends2024statistical} and randomness generation~\cite{cheng2024emergent} that weaker versions of transversal gates exist. We refer to these as {\it weak transversal gates}, which implement a set of gates in a probabilistic manner depending on the syndrome, instead of a deterministic manner (Fig.~\ref{fig:fig1}).
Both works consider surface codes of odd code distance, with the assumption that identical physical unitaries are applied to all the data qubits. Extending this framework to quantum error correcting codes with multiple logical qubits, including quantum low-density parity check (qLDPC) codes, has been an open question, as noted recently by Cheng et al.~\cite{cheng2024emergent}.
Additionally, when considering inhomogeneous application of physical unitaries, it remains entirely unknown which weak transversal gates are allowed.

An outstanding question to be answered in this context is as follows: {\it Are Pauli rotations weakly transversal in general codes, and what does this enable computationally?}
We address this question by establishing
a general construction of weak transversal Pauli rotation gates in general Calderbank--Shor--Steane (CSS) codes,
including high-rate quantum low-density parity-check (qLDPC) codes~\cite{kovalev2012improved, tillich2014quantum, zeng2019higher, breuckmann2021quantum, bravyi2024highthreshold} and topological codes such as the surface code~\cite{kitaev2003fault, fowler2012surface} and color codes~\cite{bombin2006topological, landahl2011fault}~(Sec.~\ref{sec:theory}).
We further demonstrate the impact of weak transversal gates for quantum computing in three settings: factory-free Clifford+$\phi$ architecture, postselected short-depth state preparation, and enhanced magic state distillation:
\begin{itemize}
\item {\bf Factory-free, Clifford+$\phi$ architecture (Sec.~\ref{subsec:inplace-rot},\ref{subsec:resource-est}).}
We propose a partially-fault tolerant architecture that augments the instruction set with {\it in-place} Pauli rotations implemented via a repeat-until-success (RUS) strategy. As detailed in Sec.~\ref{subsec:inplace-rot}, estimate from circuit-level  simulations show that a rotation of $0.001$ achieves a logical error rate of $9.6 \times 10^{-5}$ within $30$ rounds of syndrome extraction on a $d=7$ surface code at $\pphys=10^{-4}.$
Compared to synthesizing the same rotation from fault-tolerant primitives, this reduces the number of code cycles by a factor of ten,
and also evades the space overhead of magic state factories as well as the overhead of distilling and routing, which is in sharp contrast to existing partial FTQC architectures~\cite{choi2023fault, gavriel2023transversal, akahoshi2024partially, toshio2025practical}.
We further find in Sec.~\ref{subsec:resource-est} that resource estimates for Trotter-like circuits with $N=108$ logical qubits assuming a surface code and a [[144, 12, 12]] bivariate bicycle code~\cite{bravyi2024highthreshold} show improvements by factors of 778 and 45, respectively, in the runtime over a conventional Clifford+T architecture, enabled by the natural parallelism of rotation gates.

\item \black{
{\bf Short-depth state preparation (Sec.~\ref{subsec:stateprep}).}
Looking ahead to regimes requiring much lower logical error rates, we study shallow-depth state preparation, which underlies primitives in quantum learning, quantum simulation, and phase estimation.
By utilizing postselection, weak-transversal rotation gadgets prepare high-level equatorial magic states with a markedly improved cost-accuracy tradeoff; for example, we find that a magic state at the 11th level of the Clifford hierarchy can be prepared with up to a $10^3$-fold reduction in spacetime cost on a $d=11$ surface code at $\pphys=10^{-4}$. We also introduce a beyond-state-of-the-art protocol for preparing the phase gradient state, a catalytic resource for the phase estimation algorithm.
}

\item \black{{\bf Magic state distillation (Sec.~\ref{subsec:msd}).} We next turn to ultra-high-precision regimes where magic state distillation is indispensable.
Postselected weak transversal gates exhibit nonlinear error suppression, which can be further enhanced by segmenting the target rotation angle at the cost of increased spacetime overhead.
While segmentation is expensive as a stand-alone state preparation method, it becomes advantageous within a distillation pipeline, where initial noisy state injections are the bottleneck in terms of error but contribute only a small fraction of the total spacetime cost.
For the Meier--Eastin--Knill protocol~\cite{meier2012magic}, we numerically find that integrating segmented weak transversal gates reduces the output logical error of the $\sqrt{T}$ state by a factor of $94$ with only a modest increase in spacetime cost, which compounds even more when we perform multiple rounds of distillation.
}
\end{itemize}

\section{Theory of weak transversal operation}\label{sec:theory}
\subsection{Preliminaries and setup}

A logical quantum operation $\bar{\mathcal{U}}$ is usually said to be {\it transversal} if it can be implemented by applying the same physical gate to each physical qubit in the code blocks.
In the following, we introduce the notion of {\it strong} and {\it weak} transversality.
The former coincides with the standard definition, while the latter differs in the sense that it admits probabilistic implementation of a gate set, and hence the fault distance is usually lower than that of strong transversal gates.
For the sake of simplicity, we assume that operations are done within a single code block of an $[[n, k, d]]$ QECC, while the notion can be straightforwardly extended to involve multiple code blocks.
We first provide the definition of a strong transversal gate on a single QECC patch.
Note that we hereafter distinguish logical and physical operations by the presence or absence of the bar above the symbols, e.g., $\bar{U}$ and $U$.

\begin{definition} (Strong transversality.) \label{def:strong-transversality}
Let $\bar{U}$ be a unitary that acts on one or more code blocks, each consisting of $n$ physical qubits. A physical unitary $V$ is called a {\it strongly transversal} implementation of $\bar{U}$ if $V= \bigotimes_{i=1}^n v_i,$ where $v_i$ acts only on the $i$-th physical qubit(s) of the code block(s) involved.
\end{definition}
One of the most prominent examples of strong transversal operations is the logical Pauli operators of qubit stabilizer codes; they can all be implemented with single-qubit Pauli gates.
Other examples include the CNOT gate for Calderbank-Shor-Steane (CSS) codes~\cite{shor1996fault} and the T gate for triorthogonal codes~\cite{BravyiHaah2012}.
In these cases, the physical operation $\mathcal{V}$ is a product of local unitaries; syndrome extraction and recovery operations are performed to correct errors, but do not play a crucial role in the implementation of the logical gate. We contrast this with weak transversal gates, in which the action of physical operations themselves is not closed within the individual symmetric space defined by stabilizers. We further discuss this point after providing the definition:

\begin{definition} (Weak transversality.) \label{def:weak-transversality}
Let $\mathcal{E} = \{(p_x, \bar{U}_x)\}_x$ be an ensemble of probabilities $\{p_x|~\sum_x p_x = 1, p_x > 0\}$ and unitaries $\{\bar{U}_x\}$ that act on logical qubits of a QECC patch consisting of $n$ qubits. Assume that the following conditions are satisfied by the physical operation $\mathcal{V}$: \begin{itemize}
    \item $\mathcal{V}$ is a constant-depth sequence of physically available local operations on $n$ qubits that implements $\bar{U}_x$ with probability $p_x$.
    \item The outcome does not depend on the input logical state.
\end{itemize}
 Then, $\mathcal{V}$ is a weak transversal gate.
Here we do not require the ability to choose from the set $\{\bar{U}_x\}$; it can be a random realization.
\end{definition}

The notion of weak transversality is relevant to the projected ensembles, which are ensembles of pure quantum states generated by projective measurements on a subspace of a Hilbert space (usually a subset of qubits for a multiqubit system)~\cite{choi2023preparing, cotler2023emergent}.
Here we instead consider an ensemble of logical unitaries realized via a sequence of physical unitaries followed by projective measurements and feedback operations (see Fig.~\ref{fig:fig1}).
Note that it is crucial to require that the ensemble of unitaries $\mathcal{E}$ does not rely on an input state $\rho$, as long as $\rho$ lives in the code space.
Otherwise, we must simulate the quantum circuit in order to deduce the implemented quantum operation, which quickly becomes intractable when $n$ is large.

Due to the celebrated Eastin-Knill theorem~\cite{eastin2009restrictions, zeng2011transversality}, there is no QECC that admits a strong transversal implementation of universal gate sets, leading to the need to rely on gate teleportation~\cite{gottesman1999demonstrating} or code switching~\cite{kubica2015universal, anderson2014fault, butt2024fault} to attain universality.
Meanwhile, the Eastin-Knill theorem does not rule out {\it weak} transversal implementation.
We therefore aim to make a first attempt to derive general conditions for a QECC to admit weak transversal implementation of Pauli rotation gates.
It is useful to refer to the entire continuous gate set of $e^{i \bar{P} \theta}$ as follows.
\begin{definition} (Weak transversal Pauli rotation gate.) Let $\bar{P}$ be a nontrivial logical Pauli operator on a single QECC patch. If $\mathcal{V}_{\Phi}$ is a constant-depth sequence of physically available local operations with parameter set of $\Phi$ such that, for any $\theta$ there exists $\Phi$ that implements $e^{i \bar{P} \theta}$ with nonzero probability, then $\mathcal{V}_{\Phi}$ is a weak transversal $\bar{P}$ rotation gate.
\end{definition}
As we see in the remainder of this section, $\mathcal{V}_\Phi$ consists of local unitaries, syndrome measurements, and a recovery map. With a slight abuse of terminology, we also refer to the local unitaries as weak transversal Pauli rotation gates.

Our main theoretical result regarding the weak transversal Pauli rotation gate is Theorem~\ref{thm:general-pauli-rot}, which provides
 a sufficient condition for arbitrary CSS codes.
 To prove this, we start in Sec.~\ref{subsec:condition-single-qubit} by assuming that (i) physical unitaries are single-qubit operations, and (ii) we aim to implement single-qubit Z rotation.
 As detailed in Sec.~\ref{subsec:condition-multiqubit}, we prove Theorem~\ref{thm:general-pauli-rot} by gradually lifting the conditions.
 Namely, we consider physical unitaries acting on multiple physical qubits, and also aim to provide a method for weak transversal gates for general Pauli rotations.
 While we restrict our discussion to CSS codes, we also briefly mention the applicability to non-CSS codes in Sec.~\ref{subsec:non-css}.

\subsection{Weak transversal Z rotation for odd distance }\label{subsec:condition-single-qubit}

Our first result is the sufficient condition to admit weak transversal Z rotation on a multiqubit code block, which also applies to X and Y rotations. After stating the theorem, we provide a proof sketch, with its full derivation deferred to Appendix~\ref{app:thm1-proof}.

\begin{restatable}{theorem}{weakZthm}(Weak transversal Z rotation gate on odd-distance code)\label{thm:z-rotation}
Let a logical QECC patch be an $[[n, k, d_x, d_z]]$ CSS code. Assume that the following conditions are satisfied:
\begin{enumerate}
    \item The code distance for the logical Z operator, $d_z$, is odd.
    \item The logical Z operator of interest, $\bar{Z}_{\rm target}$, is weight $d_z$.
\end{enumerate}
Then, $U \coloneqq \prod_{j \in {\rm supp}(\bar{Z}_{\rm target})} e^{i \theta Z_j}$ is a weak transversal Z rotation gate.
\end{restatable}

\begin{proof}[Proof sketch]
The core of the theorem is to prove (i) input agnosticity of the logical operation and (ii) the ability to implement an arbitrary logical angle by tuning the physical angle. These properties ensure weak transversality and a nontrivial logical operation, respectively. In the following, we outline the proof in order.

For input agnosticity, let $\Pi_0$ be the projector onto the code space and $\{\Pi_{\bm s}\}_{\bm s} \coloneqq \{\prod_{i=1}^{n-k}\frac{1 + (-1)^{s_i}g_i}{2}\}_{\bm s}$ the projectors corresponding to a syndrome measurement result ${\bm s}$, where $g_i$ is the $i$-th stabilizer generator. After applying
$U=\prod_{j\in \mathrm{supp}(\bar Z_{\mathrm{target}})} e^{i\theta Z_j}$ and measuring the syndrome,
the Kraus operator associated with outcome $\bm s$ is
$K_{\bm s} := \Pi_{\bm s} U$.
Then,  it suffices to show that the outcome probability
$p_{\bm s}(\bar\rho)=\mathrm{Tr}[K_{\bm s} \bar\rho K_{\bm s}^\dagger]$ is independent of the input logical state $\bar\rho$.

It is convenient to expand a given logical state in terms of logical Pauli operators as
\begin{equation}
\bar{\rho} = \Pi_{\bm 0} \left(\sum_{\bar{Q}\in\{\bar{I}, \bar{X}, \bar{Y}, \bar{Z}\}} c_{\bar{Q}} \bar{Q}\right) \Pi_{\bm 0} = \frac{1}{2^k}\Pi_{\bm 0} + \sum_{\bar Q\in \{\bar{I}, \bar{X}, \bar{Y}, \bar{Z}\}^{\otimes k}\backslash \bar I} c_{\bar Q}\,\Pi_0 \bar Q .
\end{equation}
Then
\begin{equation}
p_{\bm s}(\bar\rho)=\frac{1}{2^k}\mathrm{Tr}[\Pi_{\bm s} U\Pi_0 U^\dagger]+
\sum_{\bar Q\neq \bar I} c_{\bar Q}\,\mathrm{Tr}[\Pi_{\bm s} U\Pi_0 \bar Q U^\dagger].
\end{equation}
Since $\Pi_{\bm s}$ and $\Pi_0$ are linear combinations of stabilizers, it suffices to show that
\begin{equation}
\mathrm{Tr}[S_1 U S_2 \bar Q U^\dagger]=0 \label{eq:thm1-traceless}
\end{equation}
for all $S_1,S_2\in\mathcal S$ and all nontrivial logical Paulis $\bar Q$.

If $\bar Q$ contains any logical $X$ component, then $S_1S_2\bar Q$ has nonzero $X$-support.
Since the conjugation of $S_2 \bar{Q}$ by $U$ can only append $Z$'s on the overlapping sites (it never removes $X$ support),
the term remains traceless.
Next, if $\bar Q$ is purely logical $Z$, i.e., do not contain any X or Y component, conjugation by $U$ can only create additional $Z$ operators supported on
$\mathrm{supp}(\bar Z_{\mathrm{target}})\cap \mathrm{supp}_X(S_2)$.
Because $d_z$ is odd and $S_2$ commutes with $\bar Z_{\mathrm{target}}$, this overlap is even and strictly smaller than $d_z$.
Therefore it cannot cancel a nontrivial logical-$Z$ component, whose minimum weight is $d_z$ by assumption.
Thus, $\mathrm{Tr}[S_1 U S_2 \bar Q U^\dagger]=0$ in all cases, implying that $p_{\bm s}$ is input-agnostic.

In order to show the second part, namely the capability of implementing arbitrary logical angle, we identify $\bar{U}_{\bm s}$ as a $\bar Z_{\rm target}$-rotation up to a Pauli recovery operation.
We first expand the unitary as
\begin{equation}
U=\prod_{j\in\supp(\bar Z_{\rm target})}(\cos\theta\,I+i\sin\theta\,Z_j)
=\sum_{\bm b \in \{0, 1\}^{d_z}}(\cos\theta)^{d_z-|{\bm b}|}(i\sin\theta)^{|{\bm b}|} \prod_{j\in \supp(\bar{Z}_{\rm target})} Z_j^{b_j} .
\end{equation}
Each Z term determines a syndrome $\bm s$, which
 is uniquely related to a  minimum-weight $Z$ error
$E_{\bm s}$ with  $\supp(E_{\bm s})   \subseteq\supp(\bar Z_{\rm target})$ and $w:=|E_{\bm s}|\le(d_z-1)/2$.
Moreover, $E_{\bm s}$ and $E_{\bm s}\bar Z_{\rm target}$ have the same syndrome.
Projecting onto the $s$-syndrome subspace therefore yields, on the code space,
\begin{equation}
K_{\bm s} \propto
E_{\bm s}\Big[
(\cos\theta)^{d_z-w}(i\sin\theta)^w\,\bar I
+(\cos\theta)^w(i\sin\theta)^{d_z-w}\,\bar Z_{\rm target}
\Big]\Pi_0
= \sqrt{p_{\bm s}}\,E_{\bm s}\,e^{i\bar\theta_s\bar Z_{\rm target}}\Pi_0 ,
\end{equation}
with
\begin{equation}
p_{\bm s}=
\big|(\cos\theta)^{d_z-w}(\sin\theta)^w\big|^2+
\big|(\sin\theta)^{d_z-w}(\cos\theta)^w\big|^2,
\qquad
\tan\bar\theta_s = (-1)^{\frac{d_z-1}{2}-w}\tan^{\,d_z-2w}\theta .
\end{equation}
Because $d_z$ is odd, varying $\theta$ makes $\bar\theta_s$ attain any target angle modulo a global phase, with $p_{\bm s}>0$.

\end{proof}

\subsection{General Pauli rotation} \label{subsec:condition-multiqubit}
While we have assumed only single-qubit unitary for each physical qubit in \Cref{thm:general-pauli-rot}, weak transversal operation can be even more versatile if we allow multiqubit operations.
Deferring the detailed proof to Appendix~\ref{app:thm2-proof}, here we state the Theorem and then provide a sketch of proof.

\begin{restatable}{theorem}{weakPthm} (Weak transversal implementation of general Pauli rotation gate) \label{thm:general-pauli-rot}
Let $C$ be an $[[n,k,d]]$ CSS code, and let $M\leq d$ be an odd integer. Then, for any  logical Pauli operator $\bar{P}$ of the code $C$, there exists a set of Pauli operators $\{P_m\}_{m=1}^M~(P_m\in \{I, X, Y, Z\}^{\otimes n})$ such that a unitary $\prod_{m=1}^M e^{i \theta P_m}$ followed by syndrome extraction yields weak transversal implementation of $\bar{P}$-rotation gate.
\end{restatable}

\begin{proof}[Proof sketch]

Similarly to \Cref{thm:z-rotation}, we prove the theorem in two steps; first  the input agnosticity, and second the capability of implementing arbitrary logical angle.

For the first step, it is crucial to construct the actual unitary operation for the weak transversal gate.
To accomplish this, first recall that once we choose $U$ as a product of Pauli rotations, every operator appearing in the Pauli
expansion of $U$ is a \emph{subproduct} of the rotated Paulis. Consequently, when we evaluate expressions of the form
${\rm Tr}[S_1 U S_2 \bar{Q} U^\dagger]$,
the only potentially dangerous terms are those subproducts that have a nonzero action on the code space.
If some proper subproduct acts as a logical Pauli operator (up to stabilizers), it can interfere with the logical degrees of freedom
and spoil the input-agnosticity. Thus, a natural sufficient condition for extending Eq.~\eqref{eq:thm1-traceless} is to demand that
\emph{no proper subproduct survives on the code space}.

Toward this goal, we introduce a notion called {\it disjoint support partition}. For a given target logical operator $\bar{P}$, a set of physical Pauli operators$\{P_m\}_{m=1}^M$ is a disjoint support partition if
(i) $\bar P=\prod_{m=1}^M P_m$,
(ii) $M$ is an odd number, and
(iii) for any nonempty subset $R\subsetneq [M]$, it holds that $\Pi_0\Big(\prod_{m\in R}P_m\Big)\Pi_0=0$ where $\Pi_0$ is a projection operator onto the code space.
As we detail in Appendix~\ref{app:thm2-proof}, we have constructed an efficient algorithm to find such a partition of $M\leq d$ for arbitrary logical operators. Thanks to the property (iii),  a unitary defined as
\begin{equation}
U=\prod_{m=1}^M e^{i\theta P_m}
\end{equation}
can only multiply $\prod_{m\in R \subsetneq [M]}P_m$ upon conjugating a stabilizer $S$ as $U S U^\dag$, and hence it holds that
 ${\rm Tr}[S_1 U S_2 \bar{Q}U^\dagger] = 0$ for any stabilizers $S_1$ and $S_2$ and logical Pauli operator $\bar Q$.
This indicates that $U$ followed by the syndrome extraction yields an input-agnostic logical operation, and hence $U$ is a weakly transversal gate in the sense of Definition~\ref{def:weak-transversality}.

For the second step,  we first find that any product of Pauli operators from subset of disjoint support partition $\prod_{m=1}^M P_m^{b_m}~(b_m\in\{0, 1\})$ can be uniquely related to a Pauli correction operation $E_{{\bm s}({\bm b})}$ where $\bm s({\bm b})$ indicates the syndrome invoked by the Pauli operator.
Because $M$ is odd, every ${\bm b}$ has a unique complement ${\bm b}^c = {\bm 1} \oplus {\bm b}$ with $|{\bm b}^c|=M-|{\bm b}|$.
Using (i), we have $\prod_{m=1}^M P_m^{b_m} P_m^{b_m^c} = \bar{P}_{\rm target}$ and $E_{\bm s(\bm b)}=E_{{\bm s}({\bm b}^c)}$ up to global phase, thus
\begin{equation}
\Pi_0 E_{\bm s(\bm b)}E_{{\bm s(\bm b)}}\Pi_0= \Pi_0 E_{\bm s(\bm b)}E_{\bm s(\bm b^c)}\Pi_0 = \Pi_0. \label{eq:thm2-KL-1}
\end{equation}
Next, for any pair of $\bm b$ and $\bm b'$ satisfying ${\bm s(\bm b)}\neq {\bm s}(\bm b')$, their symmetric difference  corresponds to a nonempty proper subset
$R\subsetneq[M]$, hence using (iii) we have
\begin{equation}
\Pi_0 E_{\bm s(\bm b)}E_{{\bm s({\bm b}')}}\Pi_0 = 0,    \label{eq:thm2-KL-2}
\end{equation}
Combining Eqs. \eqref{eq:thm2-KL-1} and~\eqref{eq:thm2-KL-2}, the set $\{E_{\bm s(\bm b)}\}_{{\bm b}}$ satisfies a Knill--Laflamme-type orthogonality on the code space.
Therefore stabilizer syndrome measurement can discriminate errors $E_{\bm s}$, which in principle does not discriminate the terms $\prod_m P_m^{b_m}$ and its complement $\prod_m P_m^{b_m^c} = \bar{P}_{\rm target}\cdot \prod_m P_m^{b_m}$.

The above argument shows that, for a given syndrome $\bm s$,   only the Pauli operators $\prod_{m=1}^M P_m^{b_m}$ corresponding to bitstrings $\bm b$  with $\bm s=\bm s(\bm b),~|\bm b|\leq (M-1)/2$ and its complement contribute to the logical operation. Concretely, the Kraus operator can be written using $w:=|\bm b|$ as
\begin{equation}
K_{\bm s} \propto
E_{\bm s}\Big[
(\cos\theta)^{M-w}(i\sin\theta)^{w}\,\bar I
+(\cos\theta)^{w}(i\sin\theta)^{M-w}\,\bar P_{\rm target}
\Big]
= \sqrt{p_{\bm s}}\,E_{\bm s}\,e^{i\bar\theta_{\bm s}\bar P_{\rm target}},
\end{equation}
with
\begin{equation}
p_{\bm s}=
\big|(\cos\theta)^{M-w}(\sin\theta)^{w}\big|^2+
\big|(\sin\theta)^{M-w}(\cos\theta)^{w}\big|^2,
\qquad
\tan\bar\theta_{\bm s} = (-1)^{\frac{M-1}{2}-w}\tan^{\,M-2w}\theta .
\end{equation}
Since $M$ is odd, varying $\theta$ makes $\bar\theta_{\bm s}$ attain any target angle modulo a global phase with $p_{\bm s}>0$.
Combining first and second steps proves that the circuit yields a weak transversal implementation of the $\bar P_{\rm target}$-rotation gate.
\end{proof}

\subsection{General Pauli rotation in non-CSS codes} \label{subsec:non-css}

\begin{figure}[b]
    \centering
    \includegraphics[width=0.7\linewidth]{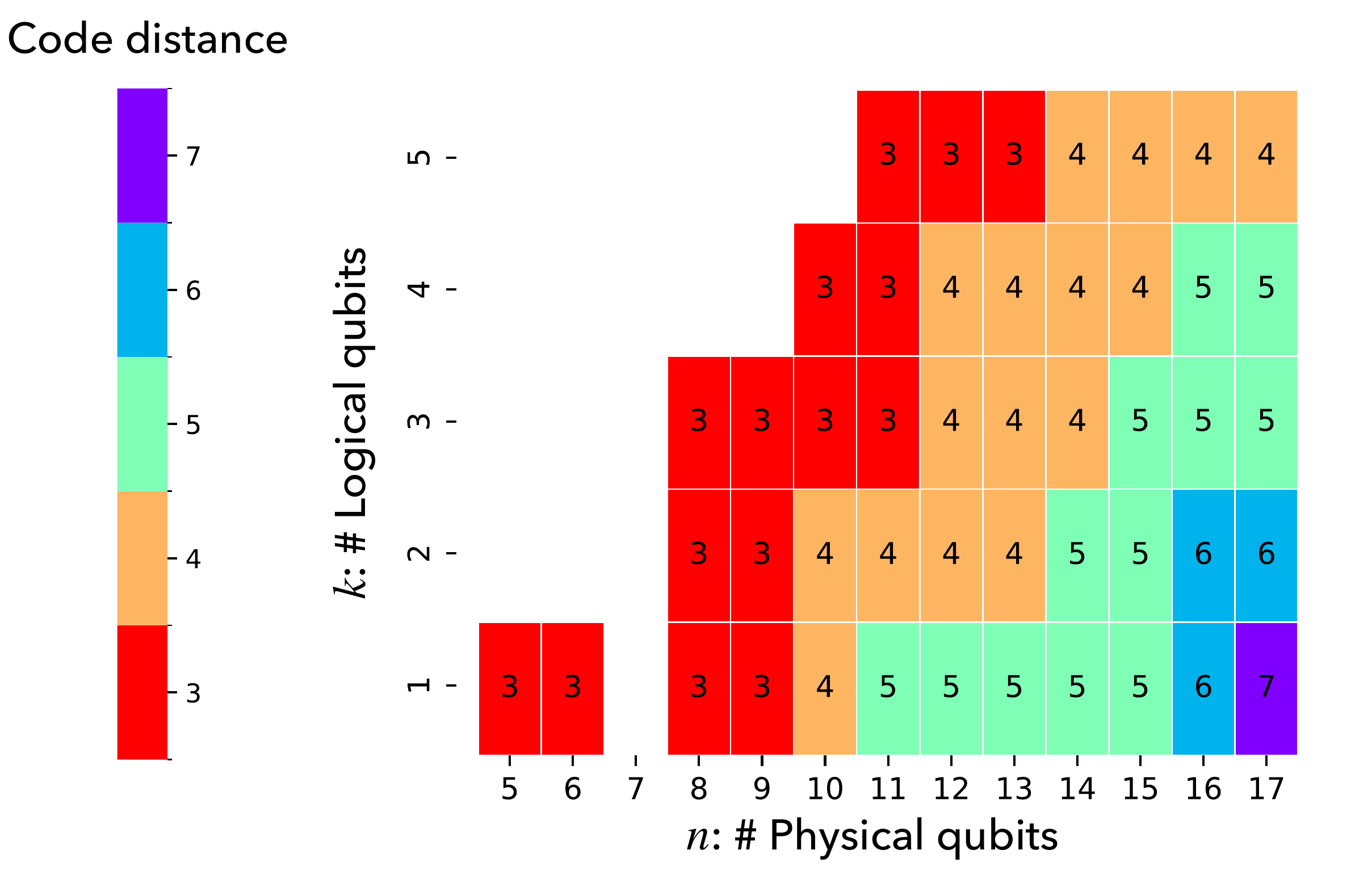}
    \caption{Numerical verification of weak transversal gates in non-CSS codes. The stabilizers are all taken from Ref.~\cite{codetable}. The color of the cells (and the inset integers) indicate the code distance. For all the non-CSS codes shown in this figure, weak transversal implementation of arbitrary Pauli rotation $e^{i \theta \bar{P}}~(\bar{P} \in \{\bar{I}, \bar{X}, \bar{Y}, \bar{Z}\}^{\otimes k})$ with $M \leq d$ is found.
    The missing cells are either CSS codes or invalid code parameters without logical operator.}
    \label{fig:non-css}
\end{figure}

In \cref{subsec:condition-multiqubit}, we have proved for CSS codes that logical Pauli rotation gates can be implemented probabilistically under some condition.
We cannot straightforwardly extend our results to non-CSS code since we have exploited the CSS property to construct disjoint support partition of logical operators; we leave it as an interesting future problem to investigate general non-CSS codes.

For future reference, here we report numerical results for small instances of non-CSS codes that are taken from Ref.~\cite{codetable}.
As shown in Fig.~\ref{fig:non-css}, we have confirmed that all non-CSS codes with $d\ge 3$ within our numerics admit weak transversal implementation  of arbitrary Pauli rotation $e^{i \theta \bar{P}}$ where $\bar{P} \in \{\bar{I}, \bar{X}, \bar{Y}, \bar{Z}\}^{\otimes k}$ using physical unitary with $M=3$ partitions.
Due to the current implementation for the calculation, we are limited to small number of logical qubits as $k \leq 5$, and also number of physical qubits as $n \leq 17$.

\section{Applications of weak transversal gates} \label{sec:applications}

We now demonstrate that weak transversal gates have versatile applications ranging from early- to long-term FTQC.
We first target the early-FTQC regime, where reducing the resource bottleneck of large Clifford+T circuits is more urgent than realizing an extremely low-error universal instruction set.
After introducing the repeat-until-success protocol to implement in-place Pauli rotation gate in Sec.~\ref{subsec:inplace-rot}, we incorporate it  into a factory-free, partially fault-tolerant architecture in Sec.~\ref{subsec:resource-est} and benchmark it against Clifford+T compilation with magic-state factories for a Trotter-like circuit.
The remaining subsections address higher-precision regimes in which state preparation becomes a central primitive.
In Sec.~\ref{subsec:stateprep}, we study short-depth preparation of high-level equatorial magic states and, in particular, the phase-gradient state, which is a catalytic resource for quantum Fourier transforms and phase-estimation routines.
We show that weak transversal gates with postselection improve the cost-accuracy tradeoff relative to state-of-the-art teleportation- or synthesis-based baselines.
In Sec.~\ref{subsec:msd}, we further show that segmented weak transversal rotations can substantially improve the output accuracy of magic state distillation with only a modest increase in circuit volume.

\subsection{Repeat-until-success protocol for in-place, addressable Pauli rotation gate}\label{subsec:inplace-rot}

Here, we aim to utilize weak transversal gates for non-fault-tolerant implementation of Pauli rotation gates.
Since the logical angles are realized randomly according to the syndrome measurement result, we consider a repeat-until-success strategy; at each step, we choose physical angles $\theta$ so that $\bar{\theta}$ is realized under some given syndrome $\bm s$.
If we obtain the desired syndrome, we terminate the iteration. If not, we again implement local unitaries, with updated physical angles so that the undesired logical angle is compensated. See Fig.~\ref{fig:rus_history}(a) for a graphical description.

Notable features of the current schemes are as follows.
\begin{itemize}
    \item {\it In-place implementation.} All the operations are closed within the logical patches and, importantly, do not require magic state factories. This implies that we can parallelize the non-Clifford gates, which are usually bottlenecked by the limited number of magic state factories.
    \item {\it Addressability.} The physical unitary is applied only to the chosen support of the target logical Pauli operator and is then followed by ordinary syndrome extraction. Therefore, even when a logical patch contains many physical qubits, the protocol can address a selected logical Pauli rotation without driving the entire patch by a strong transversal pattern.
    \item {\it Generality.} As discussed in Sec.~\ref{sec:theory}, the in-place Pauli rotation can be applied to a broad class of quantum error correcting codes. This includes many families of codes such as surface codes, color codes, and qLDPC codes that are considered to be leading candidates for fault-tolerant quantum computers.
    \item {\it On-the-fly protocol.} If the physical angles are homogeneously taken as $\theta$, then there is an analytical expression for one-to-one correspondence between the physical angle $\theta$ and the logical angle $\bar{\theta},$ as summarized in Lemma~\ref{lem:logical-theta-dist}. This is extremely helpful for the electronics, since we can calculate physical angles on-the-fly.
\end{itemize}

Since the physical angles are updated according to the syndromes, the proposed scheme is non-fault-tolerant. However, as we demonstrate in Sec.~\ref{subsec:resource-est}, in-place Pauli rotations are effective in the megaquop regime, i.e., where circuits require millions of gates under Clifford+T architecture, for two reasons. First, in such a regime, the logical error rate is moderate so that the overhead of some fault-tolerant instructions may exceed that of non-fault-tolerant ones. Second, the number of magic state factories are limited and hence parallelization by in-place, addressable operation plays a significant role.

It is useful to introduce a  convenient formula connecting the syndrome measurement result and the realized logical angle.
While we defer the proof to Appendix~\ref{app:thm2-proof}, we state the result here.
\begin{lemma}\label{lem:logical-theta-dist}
Let $C$ be an $[[n, k, d]]$ CSS code, and let $M \leq d$ be an odd integer. Let the following  be a unitary that is constructed to be a disjoint support partition of $\bar{P}_{\rm target}$,
$$U= \prod_{m=1}^M e^{i \theta X^{\bm A_m} Z^{\bm B_m}}.$$
Then, the logical operation of the unitary followed by the
syndrome measurement $\mathcal{M}$ and recovery map $\mathC$ is given as
$$\mathC \circ \mathM \circ \mathU(\bar{\rho}) = \sum_{\chi=0}^{ \frac{M-1}{2}} p_{\chi} \bar{\mathR}_{\bartheta_\chi}(\bar{\rho}),
$$ where $\bar{\mathR}_{\bartheta}(\cdot):= e^{i \bartheta \barP}(\cdot)e^{-i \bartheta \barP}$ is a logical rotation channel and $\chi$ is the weight of correction operation $C_{\bm s}$ given a syndrome $\bm s$. The probabilities and angles are given as
\begin{align}
    p_{\chi} &= \begin{pmatrix}
        M \\
        \chi
    \end{pmatrix} \left(|\cos^{M-\chi}\theta \sin^\chi \theta|^2 + |\cos^\chi \theta \sin^{M - \chi}\theta|^2\right), \label{eq:p_chi}\\
    \bar{\theta}_\chi &= \arctan \left((-1)^{\frac{M-1}{2}-\chi} \tan^{M- 2\chi}\theta\right). \label{eq:logical-theta-expression}
\end{align}
\end{lemma}

Now
we are ready to describe the repeat-until-success strategy.
Assume that one aims to implement $e^{i \barthetatarget \bar{P}_{\rm target}},$ and one has performed its weak transversal implementation that resulted in logical angle of $\bartheta.$
If  $\bartheta=\barthetatarget$, then one has succeeded in the implementation, so one can simply proceed to the next gate.
If not, one must compensate for the deviation in the next step; compute the physical angle $\theta$ using Lemma~\ref{lem:logical-theta-dist} so that $\barthetatarget - \bartheta$ is realized in the next step with as high a probability as possible. More concretely, the steps can be summarized as follows.

\begin{figure}[t]
    \centering
    \includegraphics[width=0.99\linewidth]{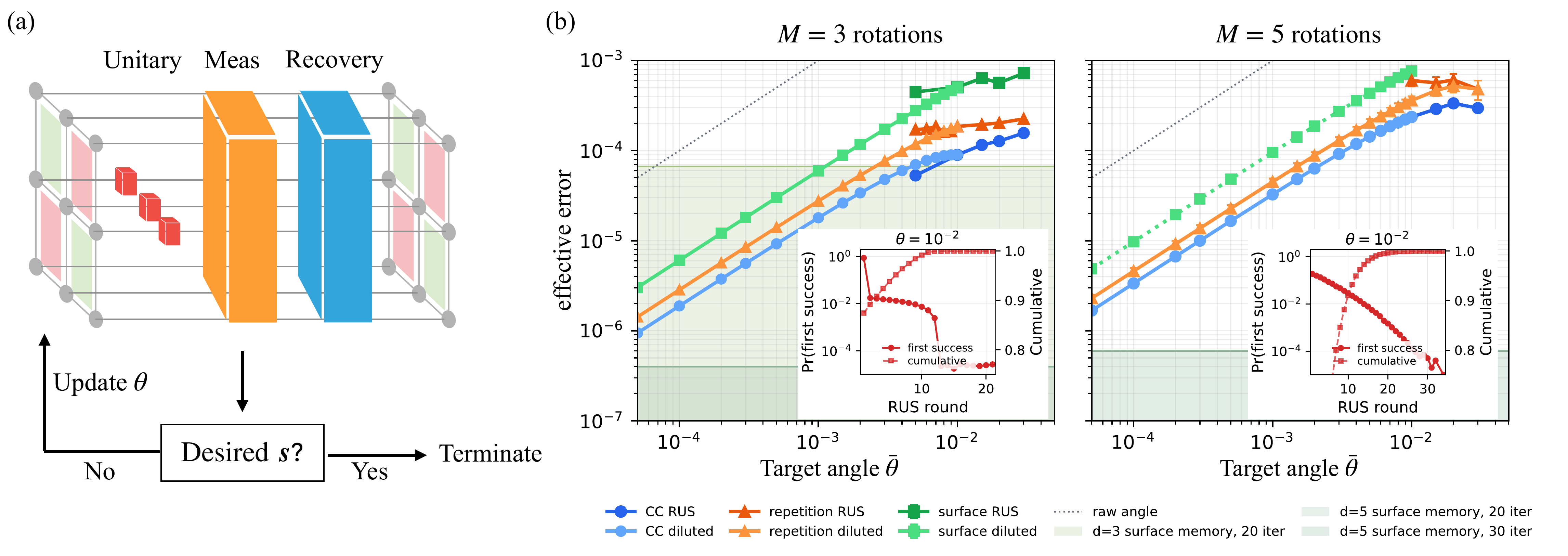}
    \caption{Repeat-until-success implementation of in-place Pauli rotations.
    (a) Schematic of the protocol.
    Physical Pauli rotations are applied on the support of the target logical Pauli operator, syndrome extraction identifies the realized logical branch, and the physical angle for the next attempt is updated unless the target angle has already been reached.
    (b) Numerical evaluation of the effective error induced by the protocol when $M=3$ and $M=5$ physical rotation gates are used, with $\pphys=10^{-4}$.
    We show ordinary RUS and diluted RUS data for the code-capacity simulation (CC) and for circuit-level simulations based on repetition and rotated surface codes.
    The insets show the history of first-success and cumulative-success probabilities  at $\bartheta=10^{-2}$.
    The dotted surface curve in the $M=5$ panel is an estimate obtained by multiplying the data of $d=M=5$ repetition code simulation data following the error rate ratio between $d=M=3$ repetition and surface codes.
    The faint and thick shaded band indicates the memory error of $d=3$ and 5 surface codes under accumulation over maximal count of iterations.}
    \label{fig:rus_history}
\end{figure}

\begin{enumerate}
    \item[0.] Initialize the cumulative logical angle as $\bar{\varphi}^{(l=0)} = 0,$ where $l$ denotes the count of iteration.
    \item[1.] Set the target logical angle as $\barthetatarget - \bar{\varphi}^{(l-1)},$ and calculate the physical angle for the $l$-th step, $\theta^{(l)}$, from Eq.~\eqref{eq:logical-theta-expression}, such that the target rotation's probability is maximized.
    \item[2.] For each ${\bm s}^{(l)}$, compute the realized angle $\bar{\theta}_{{\bm s}^{(l)}}$.
    If $\bar{\varphi}^{(l)}: = \bar{\varphi}^{(l-1)} + \bar{\theta}^{(l)}_{{\bm s}^{(l)}}$ satisfies $\bar{\varphi}^{(l)}= \barthetatarget$, then terminate the iteration. Otherwise, record the syndrome ${\bm s}^{(l)}$ and the cumulative logical angle $\bar{\varphi}^{(l)}$ for each syndrome.

    \item[3.] If terminated, compute the accumulation of Pauli operation during the protocol, and either apply the correction or record as Pauli frame to keep the state in the code space.
    If not terminated, update the iteration count as $l \leftarrow l+1$, and return to Step 1.
\end{enumerate}

If there is no hardware error, the RUS strategy systematically improves the precision as the number of attempts increases, in analogy with RUS protocols based on gate teleportation~\cite{duclos2015reducing, campbell2016_efficient}.
The relevant question is therefore how much error is induced by syndrome-extraction faults during the protocol.
We first study this for $M=3$ case using the code-capacity noise model where the syndrome extraction measurement result is flipped by $\pphys=10^{-4}$.
For a target angle $\bartheta=10^{-2}$, about $90\%$ of branches terminate in the first round, more than $99\%$ terminate within 10 rounds, and more than $99.9\%$ terminate within 20 rounds.
At this truncation depth, the induced logical operation has diamond distance below $10^{-4}$ from the target rotation.

We then compare this code-capacity estimate with circuit-level simulations using a $d=3$ bitflip code and a $d=3$ rotated surface code, as shown in Fig.~\ref{fig:rus_history}(b).
At $\bartheta=10^{-2}$, the circuit-level estimates are larger than the code-capacity value by roughly factors of two and four, respectively.
The dominant mechanism is syndrome misclassification caused by faults in the extraction circuit: a single fault can change the inferred value of $\chi$ in Eq.~\eqref{eq:logical-theta-expression}, and hence select the wrong compensation angle in the next RUS attempt.
Especially for small angles $\bartheta \ll 1$, the no-syndrome branch is misclassified into finite syndrome with probability  $\pphys M$, and results in unnecessary $\bartheta$ correction by further RUS. Although this simple argument may lead to scaling of $O(\pphys M \bartheta)$, the additional error during the preceding RUS protocol seems to yield sublinear scaling with $\bartheta$ in the plotted region.
Details of the numerical implementation and decoder choices are also given in Appendix~\ref{app:rus-numerics}.

One way to improve the logical error rate is to use mixed synthesis~\cite{hastings2016turning, yoshioka2024hunting, akibue2024probabilistic}.
Mixed synthesis probabilistically combines implementable unitary channels whose coherent mismatch from the desired target cancels to higher order.
In the present setting, we mix a noisy RUS implementation at a larger angle with an almost noiseless identity channel to synthesize a smaller target angle.
Because the weight of the larger-angle channel is proportional to the desired target angle, the stochastic contribution from RUS faults is also suppressed linearly in $\bartheta$.
The residual mismatch caused by mixing channels with different angles is much smaller than this stochastic contribution in the parameter regime considered here, so the diluted RUS protocol achieves the scaling $\mathcal{O}(\pphys \bartheta)$; see Appendix~\ref{app:mixed-synthesis}.
For $\bartheta=10^{-3}$, the $M=3$ data give ordinary-RUS effective errors of $1.2\times10^{-4}$, $1.2\times10^{-4}$, and $1.3\times10^{-4}$ for the code-capacity simulation and circuit-level simulations of repetition and surface codes, respectively.
After dilution, these values become $1.8\times10^{-5}$, $2.8\times10^{-5}$, and $6.0\times10^{-5}$.
Note that this is comparable with memory error of $d=3$ surface code ($\sim 5\times 10^{-5}$), while that for $d=5$~($\sim 5\times10^{-7}$) is significantly lower.
Repeating the diluted comparison for $M=5$ gives $3.3\times10^{-5}$, $4.5\times10^{-5}$, and $9.6\times10^{-5}$, where the last value is estimated from the ratio between repetition and surface codes in $M=3$.

It is natural to expect that the error in RUS protocol is  governed by the number of physical rotations $M$, rather than the code distance alone.
Namely, the protocol is limited by faults that corrupt the syndrome information used to identify the realized branch, rather than by the memory distance of the underlying code.
As explained in Appendix~\ref{app:scheduling-d7-surface},  the CNOT schedule of syndrome extraction can be intentionally modified so that inserting a single-qubit Pauli rotation in the middle of the extraction circuit realizes an effective multiqubit rotation, which
 makes it possible to use the same value of $M$ even in larger-distance surface-code patches.
The modified schedule intentionally lowers the fault distance for memory errors and therefore allows memory errors to accumulate more easily, but for the rotation angles considered in Fig.~\ref{fig:rus_history}(b), the RUS-induced error is still the dominant contribution.
In the remainder of this work, we use the estimated error for further numerical evaluation.

\subsection{Megaquop regime: partially fault-tolerant architecture} \label{subsec:resource-est}

\begin{figure}[b]
    \centering
    \includegraphics[width=0.99\linewidth]{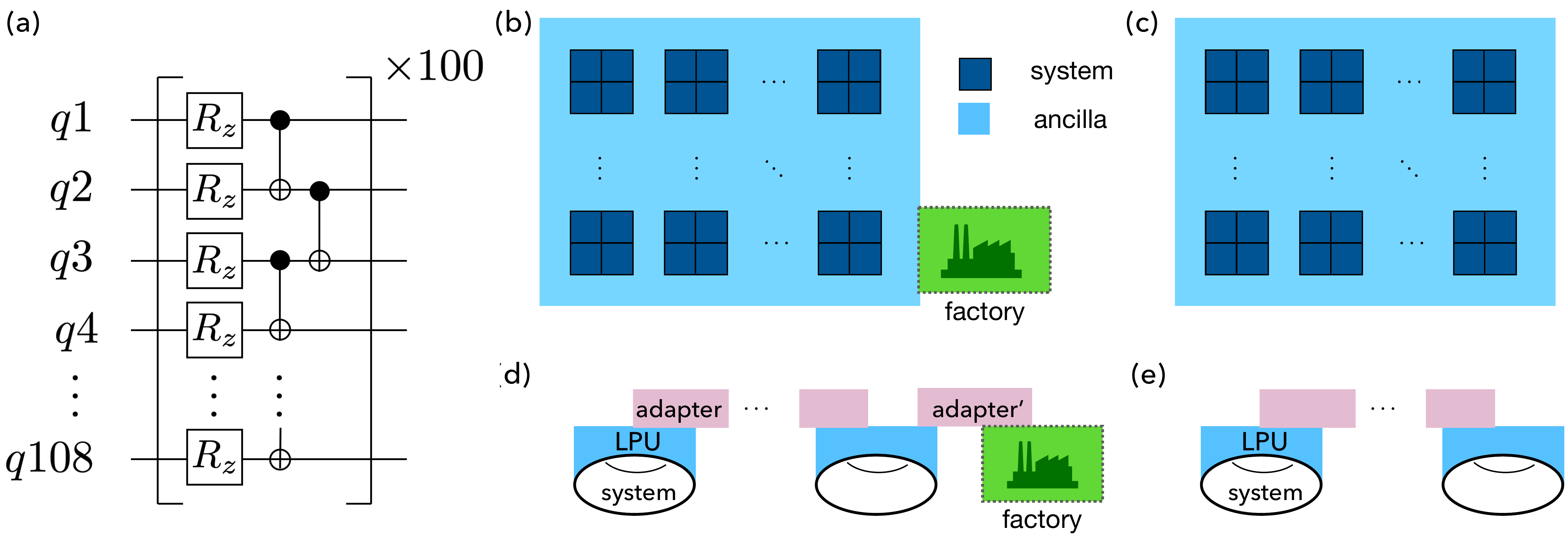}
    \caption{(a) Benchmarked quantum circuit of $N=108$ logical qubits. In each unit of gates, qubits individually undergo a $Z$ rotation with angle $0.001$ followed by two layers of CNOT gates on even and odd edges. Such a unit is repeated 100 times. For architectures, we consider (b) surface code using Clifford+T compilation with 1 to 64 magic-state factories, (c) surface code using in-place Pauli rotation, (d) [[144,12,12]] bivariate bicycle code (known as gross code) using Clifford+T compilation, and (e) gross code using in-place Pauli rotation. The code distance for the surface code is $d=7$, and that for magic state factories is $d=7$ for both surface and gross codes.
    }
    \label{fig:architecture}
\end{figure}

Early fault-tolerant quantum computers are unlikely to support the full overhead of high-throughput magic-state factories and long Clifford+T synthesis pipelines.
Therefore, partially fault-tolerant architectures  offer a natural route toward early practical quantum advantage by combining reliable Clifford operations with carefully controlled non-fault-tolerant primitives.
Here we consider Trotterized quantum dynamics as a benchmark, because (i) quantum dynamics is considered as the primary candidate for practical quantum advantage~\cite{campbell2022early, beverland2022assessing, akahoshi2025compilation}, and (ii) it contains many small-angle Pauli rotations and hence directly exposes the resource bottleneck that in-place weak-transversal rotations are designed to remove.
The performance of each architecture is quantified by the following figures of merit:
\begin{itemize}
    \item {\it Timesteps $\tau$.} Assuming that every physical operation takes the same amount of time, we call one such unit a timestep~\cite{yoder2025tour}.
    The total number of timesteps $\tau$ estimates the wall-clock time of circuit execution. Physical $R_{\rm Z}$ gates are implemented as virtual gates and hence are assumed to take no timestep~\cite{mckay2017efficient}.
    \item {\it Active Physical Qubits $N_{\rm phys}$.} If a qubit is idling during the entire execution of the circuit, we regard such a qubit to be not active. Otherwise, the qubit is active.
    We count the number of active physical qubits $N_{\rm phys}$ to quantify the spatial cost of the circuit.
    \item {\it Circuit volume $V$.} The circuit volume is defined as $V := \tau \cdot N_{\rm phys}$.
    \item {\it Total logical error rate $\bar{p}_{\rm tot}$.} When the logical error rate of the $i$-th operation is $\bar{p}_i$, we estimate the total logical error rate by the first-order sum $\bar{p}_{\rm tot}:= \sum_i \bar{p}_i.$
\end{itemize}

\begin{figure}[t]
    \centering
    \includegraphics[width=0.95\linewidth]{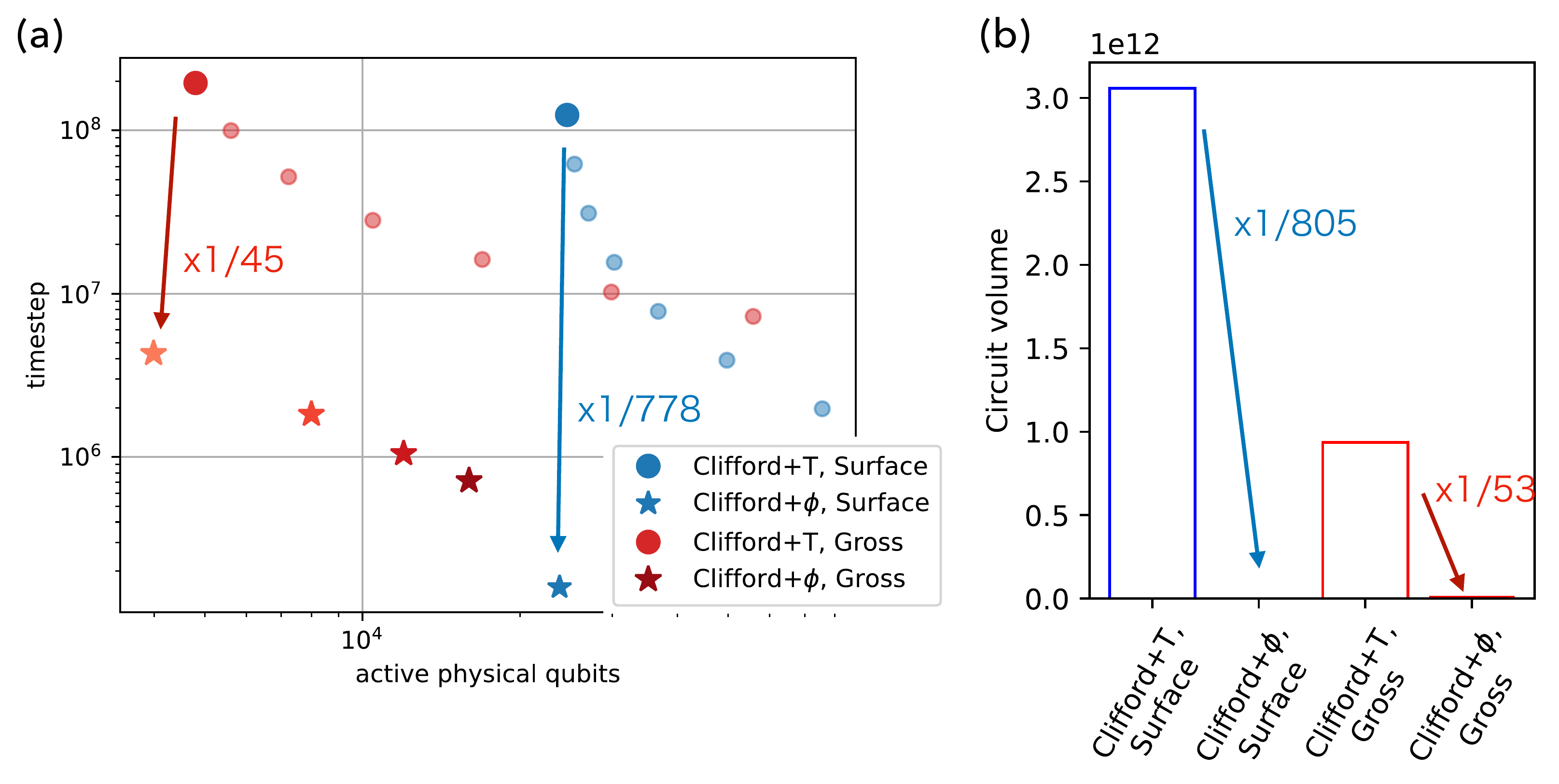}
    \caption{Resource estimation of (a) timesteps and active physical qubits and (b) the circuit volume for the circuit shown in Fig.~\ref{fig:architecture}(a).
    Compared to the Clifford+T architecture with a single magic state factory, the timesteps are reduced by factors of 780 and 45 for surface ($d=7$) and gross code, respectively, and the circuit volume is reduced by factors of 805 and 53. The faint circles indicate optimistic Clifford+T estimates with 2, 4, 8, 16, 32, and 64 magic-state factories, assuming that the runtime of synthesis can be reduced linearly with the number of factories. The data points for Clifford+$\phi$ architecture in gross code indicate that we use 10, 20, 30, and 40 modules, which enables distributed allocation of system qubits and hence more parallelized execution of CNOT gates.
    For both codes, the number of physical rotations for the weak transversal gates is $M=5$, and the physical error rate is assumed to be $\pphys=10^{-4}$.
    }
    \label{fig:resource-est}
\end{figure}

We evaluate these quantities for the four architectures shown in Fig.~\ref{fig:architecture}.
Throughout the estimates, we assume superconducting physical qubits and take the physical error rate of each operation to be $\pphys=10^{-4}$. Each Pauli rotation has angle $0.001$ and is assumed to be implemented using $M=5$ physical rotations with a logical error rate of $9.6\times10^{-5}$, following the numerical results from the previous subsection.
We choose parameters for circuit synthesis and distillation protocols so that the total error rate satisfies $\bar{p}_{\rm tot} \sim 1$, reflecting that error mitigation techniques remain available on top of error correction~\cite{tsubouchi2023universal, takagi2023universal, quek2022exponentially}.
The detailed formulas are deferred to Appendix~\ref{app:resource-est-details}; here we summarize the architectural differences.
\begin{itemize}
    \item[] {\bf Surface code,  Clifford+T.} (Fig.~\ref{fig:architecture}(b)) Clifford gates are implemented using standard surface-code operations: CNOT by lattice surgery~\cite{horsman2012surface}, Hadamard by a transversal Hadamard followed by logical patch rotation~\cite{fowler2018low, litinski2019game}, and S by lattice surgery~\cite{litinski2019game}. Non-Clifford rotations are synthesized into Clifford+T circuits, and each T gate is implemented by magic-state teleportation. We assume one level of magic-state distillation~\cite{litinski2019magic}.
    \item[] {\bf Surface code, Clifford+$\phi$.} (Fig.~\ref{fig:architecture}(c)) Clifford gates are implemented as above, while logical $\bar{R}_{\rm Z}$ gates are executed directly as in-place Pauli rotations. Each in-place rotation uses the $M=5$ gadget composed of two $R_{\rm ZZ}$ gates and three $R_{\rm Z}$ gates. Since each rotated surface-code patch encodes a single logical qubit, these logical $Z$ rotations can be applied in parallel over all patches.
    \item[] {\bf Gross code, Pauli-based computation.} (Fig.~\ref{fig:architecture}(d)) We compile the Clifford+T circuit into Pauli-based computation, where Clifford operations are absorbed by propagating Pauli operators through the Clifford circuit~\cite{bravyi2016trading,litinski2019game,yoder2025tour}. The resulting instructions are scheduled using the bicycle-code ISA, consisting of idle operations, shift automorphisms, and in-module/inter-module Pauli measurements. Each non-Clifford operation appears as a T-type Pauli rotation $e^{-i\pi P/8}$, equivalently $R_P(\pi/4)$ up to phase convention, for a possibly multi-body Pauli operator $P$, and is implemented by T-state injection consuming one magic state. We assume one level of distillation in a distance-7 surface-code magic-state factory, connected to the gross-code module through an adapter.
    \item[] {\bf Gross code, Clifford+$\phi$.} (Fig.~\ref{fig:architecture}(e)) Logical $\bar{R}_{\rm Z}$ gates are executed directly as in-place weak-transversal Pauli rotations using three $R_{\rm ZZ}$ gates and two $R_{\rm ZZZ}$ gates. CNOT gates are implemented by measurement-based Clifford operations. For in-module CNOTs, Pauli-product measurements are mediated by a pivot logical qubit; for CNOTs between different modules, the same measurement-based pattern is realized using inter-module Pauli measurements through the gross-code adapters.
\end{itemize}

The resulting resource estimates, for the circuit shown in Fig.~\ref{fig:architecture}(a), are shown in Fig.~\ref{fig:resource-est} and Table~\ref{tab:resource-est-total}.
We find that a drastic resource reduction is achieved in both surface and gross codes.
By replacing synthesized rotations with in-place rotations, the total timestep $\tau$ is reduced by a factor of 778 for the surface code and by a factor of 45 for the gross code.

For the surface code, the dominant bottleneck of Clifford+T execution is the limited supply of magic states for synthesized rotation gates. The Clifford+$\phi$ architecture resolves such a limitation, and hence reduces the number of timesteps by a factor of 778.
We also estimate an optimistic multi factory model in which the contribution of unitary synthesis for Pauli rotation  decreases linearly with the number of magic-state factories. Even under this assumption, the additional factory footprint keeps the Clifford+T estimates far above the Clifford+$\phi$ architecture in circuit volume.

For the gross code, the runtime improvement from single-factory architecture to factory-free architecture is factor of 45, which is less pronounced because the CNOT layers are comparatively expensive. This overhead originates from the fact that logical operators are not completely disjoint within one gross-code module and that in-module CNOTs are mediated by pivot logical qubits. Consequently, not all in-place rotations or pivot-mediated CNOTs can be performed simultaneously, and we schedule only mutually compatible groups in parallel. We find that such an overhead can be reduced by distributing the $N=108$ system qubits over larger number of gross-code modules. In Table~\ref{tab:resource-est-total}, the 10-, 20-, 30-, and 40-module Clifford+$\phi$ estimates restrict the maximum number of system logical qubits per module to 11, 6, 4, and 3, respectively, which increases the parallelism of the CNOT layers. Since the standard gross-code Clifford+T estimate is based on Pauli-based computation, we only report the single-factory estimate; increasing the number of factories does not translate into the same simple parallelization model as in the surface-code Clifford+T estimate. Further improvements may be possible by using additional pivot qubits or more optimized schedules. We leave such an investigation as a future work.

\input{tables/resource_estimation_total_table.tex}

\subsection{Gigaquop regime: short-depth state preparation circuits} \label{subsec:stateprep}
Short-depth state preparation is a heavily used subroutine in high-accuracy quantum algorithms.
Representative examples include GHZ/cat-like entangled states for quantum-enhanced learning~\cite{Leibfried2004, giovanetti2006quantum, ma2024learning}, tensor-network states for many-body simulation~\cite{malz2024_mpslog,smith2024_mpsconst,smith2023_aklt}, and phase-gradient states used for coherent phase rotations in quantum arithmetic and phase estimation~\cite{Gidney_2018}.
Here we use phase-gradient-state preparation as a benchmark, because it is built from a sequence of small-angle non-Clifford rotations, possibly interleaved with Clifford operations.

\input{tables/accuracy-small-angle}

We compare two conceptually different routes to small-angle rotations: teleportation-based and weak transversal methods.
Teleportation-based methods first prepare an equatorial magic state, either by angle-agnostic injection~\cite{li2015magic}, angle-agnostic injection followed by MEK distillation~\cite{meier2012magic, campbell2016_efficient}, or angle-dependent injection~\cite{choi2023fault, toshio2025practical}, and then consume it through a teleportation gadget.
These methods inherit a fallback recursion: an unfavorable measurement outcome at level $l$ requires a corrective rotation at level $l-1$, so the total gate error can be dominated by lower hierarchy levels, as in partial FTQC architectures.
Weak transversal gates take a more direct route.
Syndrome-detected attempts can simply be discarded, so postselection removes the fallback recursion and preserves the superlinear small-angle error scaling of the weak transversal gadget.
This scaling can be further improved by segmenting the target rotation into smaller angles, at the cost of additional attempts.

In the following, for the sake of comparison with prior work, we assume that logical qubits are encoded in rotated surface code.
We measure spacetime cost in what we refer to as {\it logical qubitcycles}, denoted by $V_L$.
When a logical patch has fixed size, $V_L$ is simply the sum of active logical-patch code cycles.
When the size of logical patches changes during a protocol, on the other hand, we convert the physical qubitcycles $V$ as $V_L=V/(2d^2)$, where $d$ is the final code distance.

To isolate the basic mechanism, we first consider a single Pauli rotation gate $R_{\rm Z}(\pi/2^l)$ with $l\geq 3$.
Table~\ref{tab:small-angle-subroutines} summarizes the asymptotic scaling of the relevant small-angle subroutines (see Appendix~\ref{app:small-angle-gate} for details), which is numerically demonstrated in Fig.~\ref{fig:state-prep} for $\pphys=10^{-4}$ on a $d=11$ rotated surface code.
Indeed, the angle-agnostic teleportation approach (MEK) is competitive for the lowest levels of the Clifford hierarchy, while weak transversal gates become increasingly favorable for smaller angles.
For $l=10$, a weak transversal gate with $r=10$ segments reaches an error rate comparable to MEK-distilled teleportation with roughly a $10^2$-fold smaller logical qubitcycle cost.
For $l=11$, the weak transversal gate with $r=1$ reaches a similar error scale with approximately a $10^3$-fold smaller cost.
In this small-angle regime, the postselection failure probability is only a few percent, so the cost is close to the baseline value $rd$.

\begin{figure*}[h]
    \centering
\includegraphics[width=1.\linewidth]    {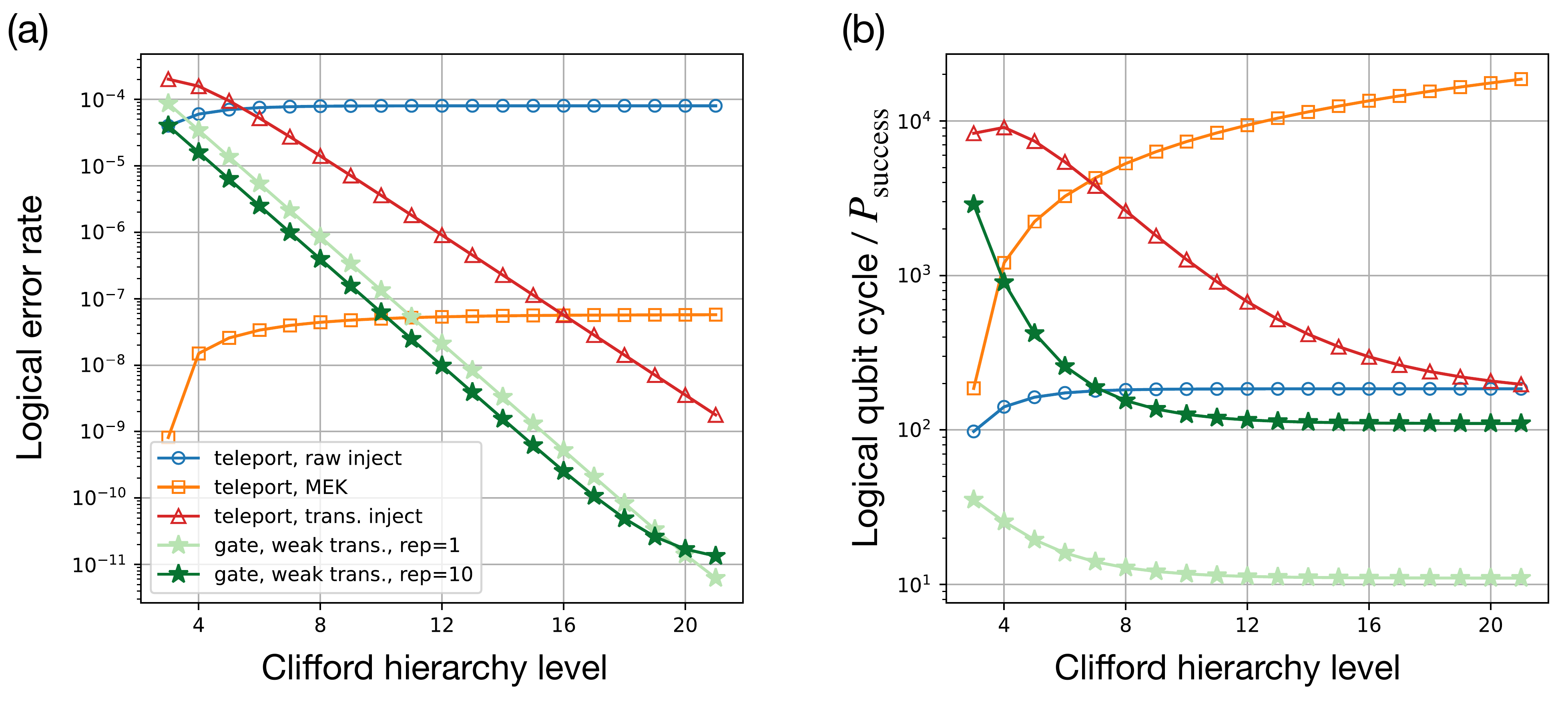}
    \caption{
Resource estimates for implementing $R_{\rm Z}(\pi/2^k)$ on a distance-$11$ rotated surface code with $\pphys=10^{-4}$.
Panel (a) compares logical error rates, and panel (b) compares logical qubitcycles per successful preparation.
For the weak transversal gate, we assume $M=3$ physical rotations.
     }
    \label{fig:state-prep}
\end{figure*}

The single-rotation benchmark is useful, but many state-preparation tasks require a collection of small rotations rather than one isolated gate.
We therefore next apply the same comparison to circuits of the form
\begin{align}
  \prod_{k=1}^K R_{\rm Z}(\theta_k) |G\rangle,
\end{align}
where $|G\rangle$ is a $K$-qubit graph state.
This form includes Trotterized Hamiltonian simulation circuits~\cite{suzuki1985general}, partial execution of instantaneous quantum polynomial circuits~\cite{bremner2011classical}, and phase-gradient state preparation.
We focus on the $K$-qubit phase-gradient state~\cite{gidney2016gradients2qfts, Gidney_2018},
\begin{align}
|{\rm PG}_K\rangle := \left(\bigotimes_{k=1}^K R_{\rm Z}\left(\frac{\pi}{2^k}\right) \right) |+\rangle^{\otimes K},
\end{align}
which can catalyze phase rotations and reduce the $T$ count of quantum Fourier transforms from $O(n^2\log 1/\epsilon)$ to $O(n^2)$.

This formulation makes the phase-gradient state a direct test of whether small-angle components should be prepared, neglected, or handled by weak transversal gates.
We compare protocols that either prepare the required equatorial magic states by MEK distillation, discard high-level components, or use a hybrid strategy.
In the hybrid strategy, MEK distillation is used up to level $K'$, while the remaining smaller rotations are prepared using segmented weak transversal gates with postselection.

\begin{figure*}[t]
    \centering
\includegraphics[width=0.65\linewidth]    {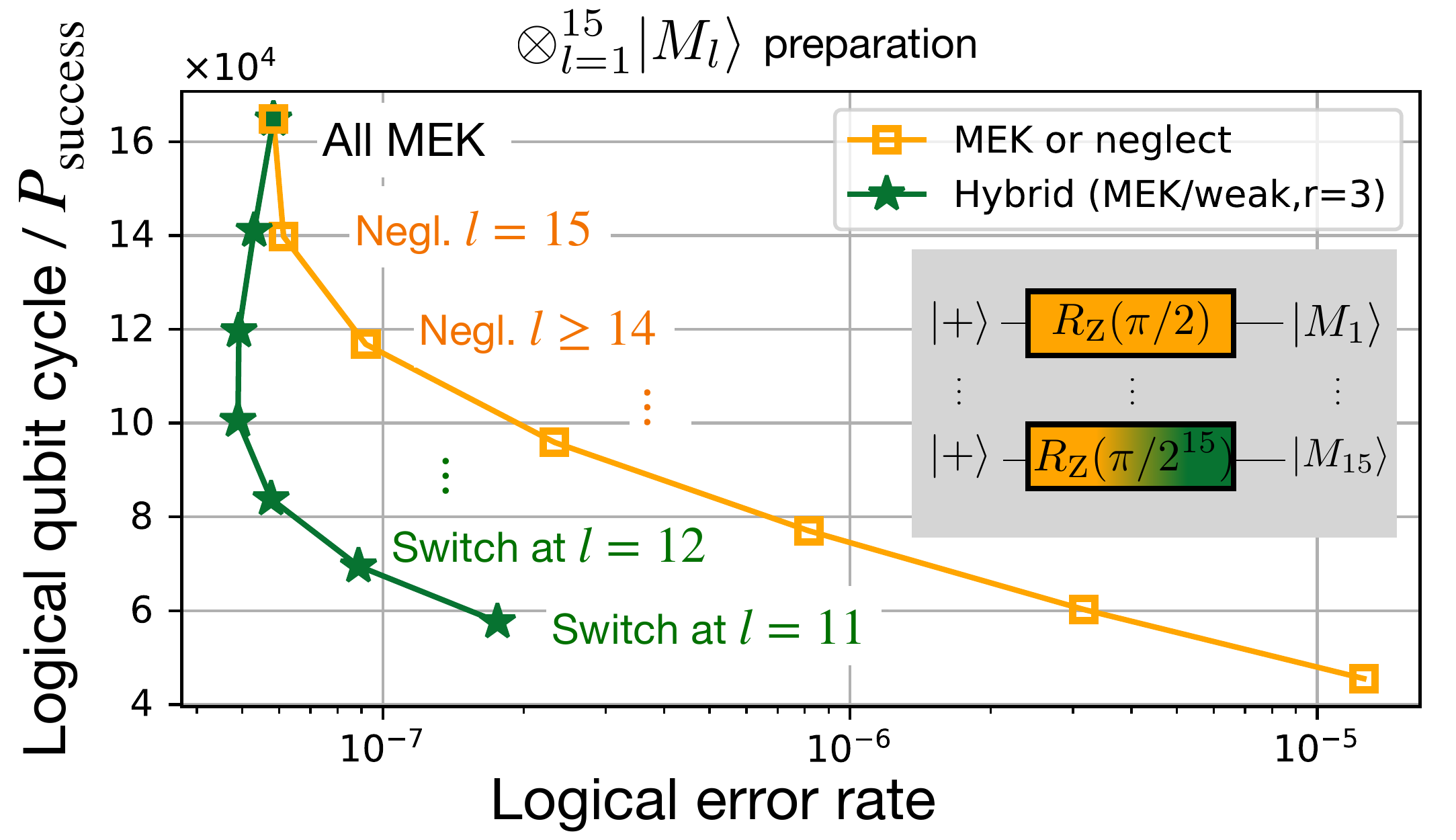}
    \caption{
Logical error and spacetime tradeoff for preparing the phase-gradient state $|{\rm PG}_{K}\rangle$ with $K=15$, assuming $\pphys=10^{-4}$ and a distance-$13$ surface code.
Orange points use MEK distillation up to a cutoff and neglect the remaining higher-level components.
Green points use hybrid protocols that combine MEK distillation for lower levels with weak transversal gates for higher levels.
}
    \label{fig:state-prep-tradeoff}
\end{figure*}

Figure~\ref{fig:state-prep-tradeoff} shows the resulting tradeoff for $K=15$.
The data show that increasing the number of weak-transversal segments is not uniformly optimal.
For larger angles, segmentation can be costly because the lower success probability increases the expected volume; for smaller angles, the same mechanism gives a favorable error-volume tradeoff.
The hybrid protocol also improves over the simple strategy of neglecting the highest-level phase-gradient components~\cite{Gidney_2018}: weak transversal gates make those small-angle components inexpensive enough to include, improving accuracy without the full cost of MEK distillation.

\begin{figure*}[b]
    \centering
\includegraphics[width=0.85\linewidth]    {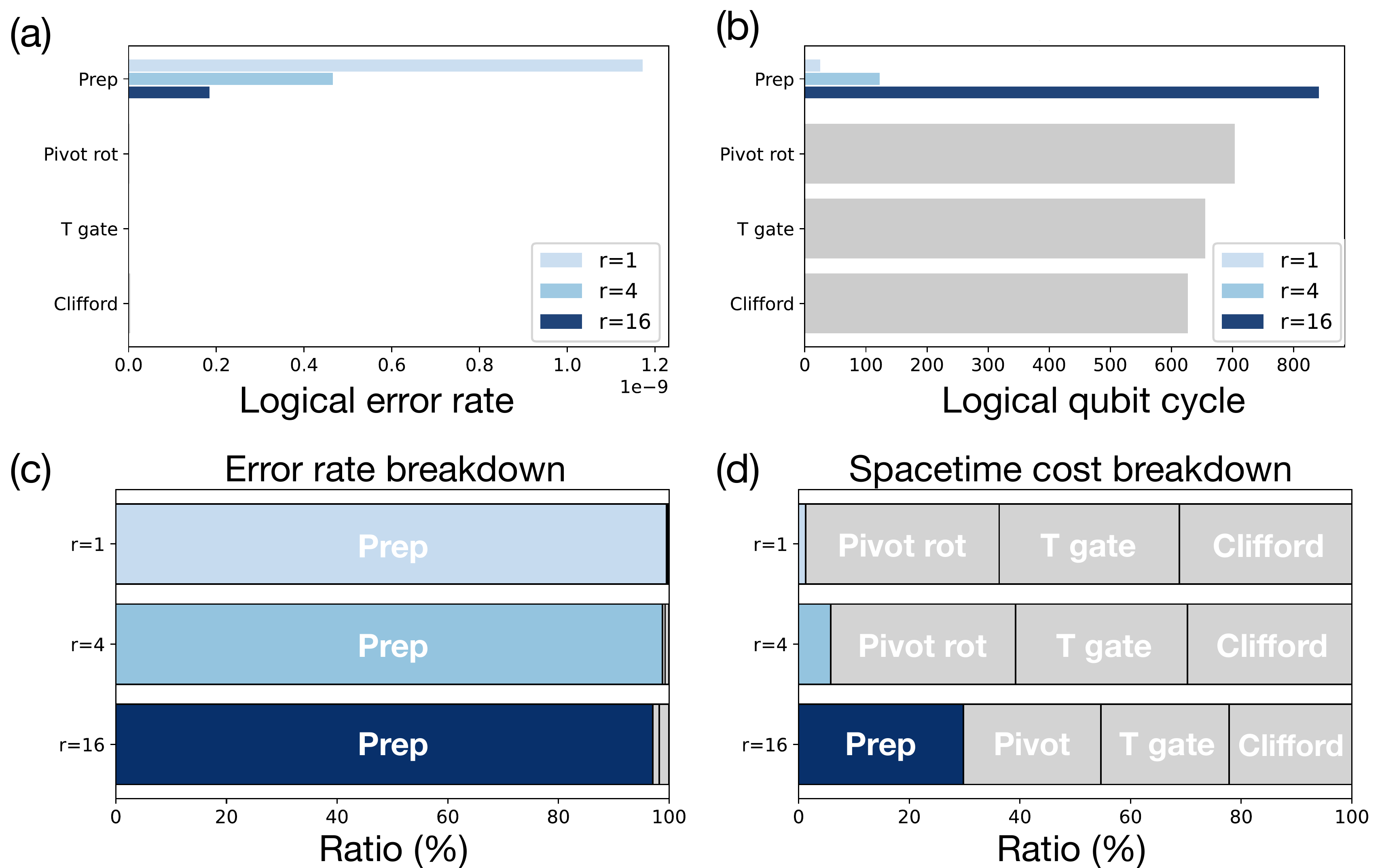}
    \caption{
Breakdown of logical error rate and spacetime cost for magic state distillation of $|M_{l=4}\rangle$. Here, we assume that T gates and the pivot rotation gate (which is a T gate) are already distilled~\cite{hirano2024leveraging}. The physical error rate is assumed to be $\pphys=10^{-4}$, and logical qubits are individually encoded via surface code of code distance $d=11$.
     }
    \label{fig:breakdown_resource}
\end{figure*}
\subsection{Teraquop regime: accuracy-enhanced magic state distillation}
\label{subsec:msd}

Weak transversal gates play a significant role even in the ultra high precision regime that tolerates a trillion gate count in terms of Clifford+T gate set.
One of the most practical examples is the magic state distillation.
The success in the state preparation, as shown in Sec.~\ref{subsec:stateprep}, naturally implies that weak transversal gates can be used to improve the accuracy of magic state injection.
Furthermore, we may benefit from the segmenting the rotation angles, especially when the bottleneck in the precision is in the rotation gate while the spacetime cost is dominated by a different process.

\begin{figure*}[b]
    \centering
\includegraphics[width=0.9\linewidth]    {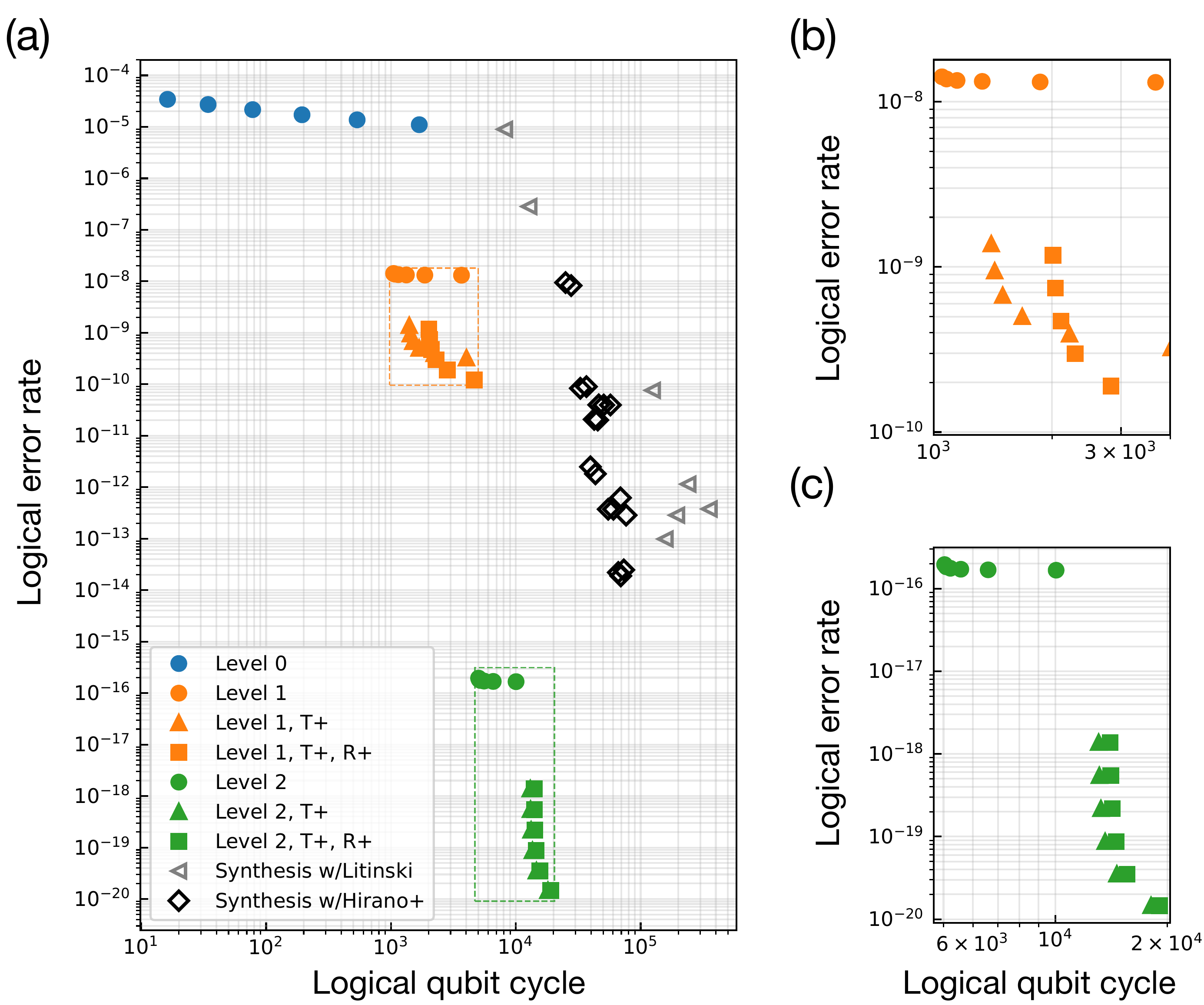}
    \caption{
Tradeoff between spacetime cost in units of logical qubit cycle and the logical error rate for distillation of $\sqrt{T}$ state. Compared with the Clifford+T-synthesis-based preparation using state-of-the-art T-state distillation protocol~\cite{ross2014optimal, litinski2019magic, hirano2024leveraging}, our proposal achieves orders of magnitude more precise magic state preparation. When we prepare the initial state, we take the number of segmentation as
$r = 1, 2, 4, 8, 16, 32.$ The number of distillation rounds by the MEK protocol is indicated by colors of blue (raw), orange (1 round), and green (2 rounds). Physical error rate is assumed to be $\pphys = 10^{-4}$ for all data.
     }
    \label{fig:distillation-resource}
\end{figure*}

One representative situation is indicated in Fig.~\ref{fig:breakdown_resource} where we show the breakdown of contribution to the spacetime cost and logical error rate to the MEK distillation protocol.
Here, we aim to distill an equatorial magic state from the $l=4$-th level of the Clifford hierarchy, which is known to nearly halve the non-Clifford gate count in synthesizing a single-qubit gate, compared to the case when only T gates are used~\cite{Kliuchnikov2023shorterquantum}.
The contributions are classified into four categories: the magic state injection, pivot rotation gate, T gate, and Clifford gates.
As we can see from results for $r=1$ in panels (a) and (c), the magic state injection dominates the error of the output state, while panels (b) and (d) indicate that the spacetime cost is only a few percent of the entire protocol.
By segmenting the rotation angle into $r=16$ steps, we can achieve a 6-fold reduction of logical error rate from $1.2\times10^{-9}$ to $1.9\times 10^{-10}$, with a moderate increase in spacetime cost from $2011$ to $2827$ in terms of the logical qubit cycle.

We gain even more by performing more rounds of distillation (see Fig.~\ref{fig:distillation-resource}). Specifically, we find that the aforementioned bottleneck structure is also present when we perform two rounds of magic state distillation, and hence, by segmenting weak transversal gates into $r=32$ steps, we can achieve a 94-fold reduction of logical error rate from $1.4\times 10^{-18}$ to $1.5\times 10^{-20}$, with an increase in the spacetime cost from $1.4\times 10^4$ to $1.9\times 10^{4}$ in terms of logical qubit cycle. We remark that such a preparation protocol is far more efficient than naively synthesizing the magic state. As shown in Fig.~\ref{fig:distillation-resource}(a), synthesizing the T gates consumes nearly 10 times the spacetime cost even when we rely on state-of-the-art T-state preparation methods~\cite{litinski2019magic, hirano2024leveraging} using a nearly optimal synthesis algorithm~\cite{ross2014optimal}.
For instance, $4\times 10^{3}$ logical qubit cycles are required in order to achieve a logical error rate of $10^{-10}$ using one round of the MEK protocol, while synthesis consumes $3\times 10^4$ logical qubit cycles.
With a similar spacetime cost, our protocol can distill down to a logical error rate of $10^{-20}.$

\section{Discussion} \label{sec:conclusion}

In this work, we have developed a general theory of weak transversal gates---logical operations implemented by local physical unitaries for probabilistic implementation of a set of logical unitaries rather than deterministic one---and shown how they expand the scope of transversal techniques for quantum computation on logical qubits.
Our theoretical contribution is to show the sufficient condition to admit weak transversal implementation of multiqubit Pauli rotation gates, together with a constructive method to realize them in CSS codes.
To provide quantitative validation, we have considered three tasks: factory-free, partially fault-tolerant architecture (Sec.~\ref{subsec:inplace-rot},~\ref{subsec:resource-est}), short-depth state preparation (Sec.~\ref{subsec:stateprep}), and enhanced magic state distillation (Sec.~\ref{subsec:msd}).
Collectively, these results position weak transversality as a practical, scalable complement to deterministic transversal gates, well-suited to early FTQC settings that must tolerate qubit dropout and dynamically changing codes.

Beyond the in-place implementation of Pauli rotation, the application of weak transversal gates are envisioned as follows.
\begin{itemize}
    \item {\bf Scrambling in logical qubits.} Our work directly solves the open question cast in Ref.~\cite{cheng2024emergent} regarding the many-body generalization of emergent randomness generation.
    \item {\bf Dynamical phase transition.} Recent findings show that dynamics of quantum circuit fosters rich physics in terms of competition between measurement, entanglement, and magic~\cite{skinner2019measurement, li2019measurement, bejan2024dynamical}. The framework of weak transversal gates provides an excellent platform to simulate such emergent phenomena.
    \item {\bf Resource distillation.} Beyond magic state, numerous quantum resource preparation including entanglement and coherence leverage distillation protocols. Our scheme is directly applicable to such tasks which is indispensable for quantum communication or distributed quantum computation.
    \item {\bf Teleportation-free architecture in flying qubits.} The in-place Pauli rotation gates opens a possibility to completely eliminate logical ancilla qubits for encoded flying qubits which can perform entangling operations in a transversal manner. Furthermore, combined with code concatenation, we may flexibly tune the achieved error scaling for both the transversal and weak transversal operations.
\end{itemize}

We finally remark on future problems.
First, it is interesting to investigate whether Theorem~\ref{thm:general-pauli-rot} can be extended to non-CSS stabilizer codes, an open problem also mentioned in Sec.~\ref{sec:theory}.
Second, it is important to evaluate the performance of in-place Pauli rotation based on circuit-level noise simulation beyond the currently executed size. The major bottleneck is the coexistence of non-Clifford gates and stochastic noise with multiple syndrome extraction. Large-scale simulation remains challenging even with  current classical simulation techniques~\cite{darmawan2017tensor, miller2025efficient}, and hence further development is required.
Third, it is intriguing to explore whether one can gain from inhomogeneous physical rotation angles. It has been found in the context of strong transversal gates that some non-Clifford gates are implemented via inhomogeneous local gates~\cite{Campbell2016SICC, GarvieDuncan2018SICC}, and we foresee that flexibility of weak transversal gates can be leveraged in a similar way.

\section*{Acknowledgments}
The authors wish to thank Sergey Bravyi,  Michael Gullans, Ying Li, Patrick Rall, Shiro Tamiya for fruitful discussions.
We are grateful to Yutaka Hirano for his generosity for sharing numerical data on (0+1)-level distillation.
N.Y. is supported by JST Grant Number JPMJPF2221, JST CREST Grant Number JPMJCR23I4, IBM Quantum, Google Quantum AI, JST ASPIRE Grant Number JPMJAP2316, JST ERATO Grant Number JPMJER2302, and JST [Moonshot R\&D] [Grant Number JPMJMS256J]. We acknowledge the use of large language models to improve the clarity of original text.
This research used resources of the Oak Ridge Leadership Computing Facility, which is a DOE Office of Science User Facility supported under Contract DE-AC05-00OR22725.
\\

{\it Note added.---} Upon preparation of the manuscript, an independent work~\cite{huang2025robust} appeared on arXiv that proposed to adaptively implement in-place Pauli rotation on a rotated surface code.
Physical rotations are applied to the entire code patch in Ref.~\cite{huang2025robust} and hence it cannot be generalized to multiqubit logical patch, while our scheme is capable of general CSS codes.

\clearpage
\appendix

\section{Proof of Theorem 1} \label{app:thm1-proof}

In this Appendix, we provide a detailed proof for \Cref{thm:z-rotation}. For the sake of convenience, we restate the theorem.
\weakZthm*

\vspace{2mm}
\noindent

To prove \Cref{thm:z-rotation}, it is important to keep in mind that if a quantum operation corresponds to a quantum channel on the logical space, then the action on the logical space should not depend on the input state. As noted in Ref.~\cite{cheng2024emergent}, the following well-known fact is crucial:
\begin{lemma}\label{lem:fact}
    Let a quantum channel $\mathcal{M}(\cdot) = \sum_x K_x\cdot K_x^\dagger$ be specified by a set of Kraus operators $\{K_x\}$.
    The probability distribution $p_x = {\rm Tr}[K_x \rho K_x^\dagger]$ is independent of the input state $\rho$ if and only if $K_x = \sqrt{p_x} U_{x}$.
\end{lemma}
In the context of weak transversal gate, the set of Kraus operators are $\{\sqrt{p_{\bm s}} \Pi_{\bm s} U\}_{\bm s}$, where $p_{\bm s}$ is the probability of obtaining the syndrome $\bm s$ and $\Pi_{\bm s} := \prod_{i=1}^{n-k} \frac{1 + (-1)^{s_i} g_i}{2}$ is the projection operator defined by the set of stabilizer generators $\{g_i\}_{i=1}^{n-k}$.
In order to bridge  Lemma~\ref{lem:fact} and Theorem~\ref{thm:z-rotation}, we first prove the following useful lemma.

\begin{lemma} \label{lem:logical_cancel}
    Let $\mathcal{S}$ be a stabilizer group that specifies the code space of an $[[n, k, d_x, d_z]]$ CSS code.
    Assume that conditions in \Cref{thm:z-rotation} are satisfied.
    Then, for any  $S_{1}, S_2 \in \mathcal{S}$, and $\bar{Q}\in \{\bar{I}, \bar{X}, \bar{Y}, \bar{Z}\}^{\otimes k}\backslash \{\bar{I}\}$, a unitary $U = \prod_{j \in {\rm supp}(\bar{Z}_{\rm target})}e^{i \theta_j Z_j}$ satisfies
    \begin{equation}
        {\rm Tr}[S_1 U S_2 \bar{Q} U^\dagger] = 0.\label{eq:traced_conjugation}
    \end{equation}
\end{lemma}
\begin{proof}
     We first prove Eq.~\eqref{eq:traced_conjugation} for logical operators that include at least one logical X operator, and then consider those without any logical X operator.
    For the sake of description, we utilize the symplectic representation of Pauli operators, while global phase is neglected since it does not affect the discussion after all.
    We denote a stabilizer as $S_{m} = X^{{\bm a}_m}Z^{{\bm b}_m}~({\bm a}_m, {\bm b}_m\in \mathbb{F}_2^n)$, and a single logical operator as $\bar{X}_i = X^{{\bm \alpha}_i}~({\bm \alpha}_i\in\mathbb{F}_2^n)$ and $\bar{Z}_i = Z^{{\bm \beta}_i}~({\bm \beta}_i \in \mathbb{F}_2^n)$.
    For a product of logical operators, we indicate them as $\bar{Q} = \bar{X}^{\bm \mu} \bar{Z}^{ \bm \nu} = X^{\sum_i \mu_i {\bm \alpha}_i}Z^{\sum_i \nu_i {\bm \beta}_i}~({\bm \mu},{\bm \nu}\in\mathbb{F}_2^k).$

    First, assume that the nontrivial logical X operator is present in $\bar{Q}$, i.e., ${\bm \mu} \neq {\bm 0}$.
    It is illustrative to first assume $\bm \nu = {\bm 0}$. Here, the product of stabilizers and a logical operator is written as
    \begin{eqnarray}
        S_1 S_2 \bar{Q} =  S_1 S_2 \bar{X}^{\bm \mu} &=& X^{{\bm a}_1+{\bm a}_2 + \sum_i \mu_i{\bm \alpha}_i} Z^{{\bm b}_1+{\bm b}_2}. \label{eq:stabilizer_XL}
    \end{eqnarray}
    Since any nontrivial logical operator is a centralizer of $\mathcal{S}$, we have ${\rm Tr}[S \bar{X}^{{\bm \mu}\neq {\bm 0}}]=0$ for any $S\in\mathcal{S}$.
    To prove Eq.~\eqref{eq:traced_conjugation}, we inspect how Eq.~\eqref{eq:stabilizer_XL} is modified under the conjugation by $U$.
    Note that for single-qubit Pauli operators we have
    \begin{align}
     e^{i \theta Z} X Z^be^{- i\theta Z} &= \cos (2\theta) X Z^b-i \sin(2\theta) X Z^{b+1} \nonumber \\
     &= \sum_{w=0, 1} c_{w} X Z^{b+w},
     \label{eq:pauli-conjugation}
    \end{align}
    where $c_w \in \mathbb{C}$ is a coefficient.
    We can easily derive a multiqubit version  by noting that the conjugation acts nontrivially only for the overlapping supports between $U$ and $X^{\bm a}$:
    \begin{eqnarray}
        \left(\prod_{j\in{\rm supp}(\bar{Z}_{\rm target})} e^{i  \theta_j Z_j} \right) X^{\bm a}Z^{\bm b}\left(\prod_{j\in{\rm supp}(\bar{Z}_{\rm target})} e^{-i  \theta_j Z_j} \right) = \sum_{{\bm w}\in W({\rm supp}(\bar{Z}_{\rm target})\cap {\rm supp}(X^{\bm a}))} c_{\bm{w}} X^{\bm a}Z^{{\bm b}+ {\bm w} \land {\bm a}}, \nonumber \\
        \
    \end{eqnarray}
    where $W({\Lambda}) = \{\sum_{j\in \Lambda} f_j {\bm e}_j| f_j\in \{0, 1\}\}$ with $e_j \in \mathbb{F}_2^n$ indicating a vector with only the $j$-th element being unity.
    Here, we have introduced $\land$ for the bit-wise AND operation. This allows us to compute the conjugation of Eq.~\eqref{eq:stabilizer_XL} as
    \begin{align}
S_1 U S_2 \bar{X}^{\bm \mu} U^\dagger &=
X^{{\bm a}_1}Z^{{\bm b}_1} \cdot U \cdot  X^{{\bm a}_2} Z^{{\bm b}_2} \cdot X^{\sum_i \mu_i{\bm \alpha}_i} \cdot U^\dagger
\\
&= X^{{\bm a}_1} Z^{{\bm b}_1} \cdot Z^{{\bm b}_2} \cdot \sum_{{\bm w} \in W_{\Lambda'}} c_{\bm{w}} X^{{\bm a}_2+\sum_i \mu_i{\bm \alpha}_i} Z^{{\bm w} \land({\bm a}_2+\sum_i \mu_i{\bm \alpha}_i)}\\
&=X^{{\bm a}_1+{\bm a}_2 + \sum_i \mu_i {\bm \alpha}_i} \sum_{{\bm w} \in W_{\Lambda'}} c_{\bm{w}}  Z^{{\bm b}_1+{\bm b}_2 + {\bm w} \land ({\bm a}_2+\sum_i \mu_i {\bm \alpha}_i)}, \label{eq:XL_conjugation}
    \end{align}
    where we have introduced $\Lambda' := {\rm supp}(U)\cap{\rm supp}(X^{{\bm a}_2+\sum_i \mu_i{\bm \alpha}_i})$.
    Since ${\bm \mu} \neq {\bm 0}$, it holds that ${\bm a}_1 + {\bm a}_2 + \sum_i \mu_i {\bm \alpha}_i \neq {\bm 0}$, and thus the trace of Eq.~\eqref{eq:XL_conjugation} must be zero as
    \begin{eqnarray}
        {\rm Tr}[S_1 U S_2 \bar{X}^{\bm \mu} U^\dagger] = 0. \label{eq:logical-X-trace-zero}
    \end{eqnarray}
    Since Eq.~\eqref{eq:logical-X-trace-zero} follows from the argument regarding the support of X operators,
    Eq.~\eqref{eq:logical-X-trace-zero} holds for ${\bm \nu} \neq {\bm 0}$ as long as ${\bm \mu} \neq {\bm 0}$.

    Next, we consider the case when $\bar{Q}$ is composed only of logical Z operators, i.e., ${\bm \mu} = {\bm 0}$ and ${\bm \nu} \neq {\bm 0}$.
    Similar to Eq.~\eqref{eq:XL_conjugation}, the conjugation of  stabilizer operators can be written as
    \begin{eqnarray}
        S_1 U S_2  \bar{Z}^{\bm \nu} U^\dagger = X^{{\bm a}_1+{\bm a}_2}\sum_{{\bm w} \in W({\rm supp}_Z(U) \cap {\rm supp}_X(S_2))} c_w  Z^{{\bm b}_1 + {\bm b}_2 + {\bm w} \land {\bm a}_2 + \sum_i \nu_i {\bm \beta}_i},
    \end{eqnarray}
    where we have introduced ${\rm supp}_X$ and ${\rm supp}_Z$ to denote the support of X and Z components, respectively.
    Using $|\cdot|$ to denote the Hamming weight, it holds that
    \begin{align}
        \max_{\bm w} |{\bm w} \land {\bm a}_2| = \max_{\bm w} |{\bm w}| =|{\rm supp}_Z(\bar{Z}_{\rm target}) \cap{\rm supp}_X(S_2)| \leq |{\rm supp}_Z(\bar{Z}_{\rm target})|=d_z,\label{eq:hamming_weight_overlap}
    \end{align}
    while the equality in $|\supp_Z(\bar{Z}_{\rm target})\cap \supp_X(S_2)| \leq |\supp_Z(\bar{Z}_{\rm target})|$ cannot be satisfied when $d_z$ is odd.
This follows from the fact that $|{\rm supp}_Z(\bar{Z}_{\rm target}) \cap {\rm supp}_X(S_2)|$ is an even number when $[\bar{Z}_{\rm target}, S_2]=0$ holds.
    On the other hand, since $S_1 S_2 \bar{Z}^{\bm \nu}$ is a nontrivial logical Z operator for any $S_1, S_2 \in \mathcal{S}$, the Hamming weight satisfies
    \begin{align}
    \min_{{\bm b}_1, {\bm b}_2 \in G_{Z}}{\rm wt}({\bm b}_1 + {\bm b}_2 + \sum_i \nu_i {\bm \beta}_i)    \geq d_z,\label{eq:hamming_weight_logical}
    \end{align}
    where $G_Z$ is the generator matrix for Z stabilizers.
    Due to Eqs.~\eqref{eq:hamming_weight_overlap},~\eqref{eq:hamming_weight_logical}, we find that $Z^{{\bm w} \land {\bm a}_2}$ and $Z^{{\bm b}_1 + {\bm b}_2 + \sum_i \nu_i {\bm \beta}_i}$  do not cancel out with each other; their product always yields a nontrivial Pauli Z operator, and hence its trace is zero. Therefore, for any $S_1, S_2\in \mathcal{S}, $
    \begin{align}
        {\rm Tr}\left[S_1 U S_2 \bar{Z}^{\bm \nu} U^\dag\right] = 0.
    \end{align}

    Combining the arguments, we conclude that the following holds for any nontrivial logical operator $\bar{P}\in \{\bar{I}, \bar{X}, \bar{Y}, \bar{Z}\}^{\otimes k}\setminus \{\bar{I}\}$,
    \begin{eqnarray}
        {\rm Tr}[S_1 U S_2 \bar{P} U^\dagger] = 0.
    \end{eqnarray}
\end{proof}

We further proceed by noting that any Kraus operator for stabilizer measurement can be decomposed into sum of stabilizer operators as
\begin{eqnarray}
    \Pi_{\bm s} = \prod_{i=1}^{n-k} \frac{I + (-1)^{s_i} g_i}{2} = \sum_{S \in \mathcal{S}} f_S S.
\end{eqnarray}
By straightforwardly combining this expression with \Cref{lem:logical_cancel}, we obtain the following:

\begin{lemma}\label{lem:logical_cancel2}
    Assume that a CSS code satisfies the conditions given in \Cref{thm:z-rotation}. Then, for all $\bar{Q} \in \{\bar{I}, \bar{X}, \bar{Y}, \bar{Z}\}^{\otimes k} \setminus \{\bar{I}\}$, a unitary $U = \prod_{j \in \supp(\bar{Z}_{\rm target})} e^{i \theta Z_j}$ satisfies
    \begin{eqnarray}
        {\rm Tr}[\Pi_{\bm s} U \Pi_{\bm 0} \bar{Q} U^\dagger] = 0.
    \end{eqnarray}

\end{lemma}

Building on top of \cref{lem:logical_cancel2}, we  show that the stabilizer measurement results do not depend on the input logical state. Note that any logical state $\bar{\rho}$ can be decomposed into sum of logical Pauli operator as
\begin{eqnarray}
    \bar{\rho} = \Pi_{\bm 0} \left(\sum_{\bar{Q}\in \{\bar{I}, \bar{X}, \bar{Y}, \bar{Z}\}^{\otimes k}} c_{\bar{Q}} \bar{Q} \right)\Pi_{\bm 0}= \frac{1}{2^{k}}\Pi_{\bm 0}
    + \sum_{\bar{Q}\in\{\bar{I}, \bar{X}, \bar{Y}, \bar{Z}\}^{\otimes k} \backslash \{\bar{I}\}} c_{\bar{Q}} \Pi_{\bm 0} \bar{Q},
\end{eqnarray}
where $c_{\bar{Q}} \in \mathbb{R}$ is the Pauli coefficient that encodes the logical information of the state.
Then, the measurement probability distribution $p_{\bm s}$ can be computed using \Cref{lem:logical_cancel2} as
\begin{eqnarray}
    p_{\bm s} = {\rm Tr}[\Pi_{\bm s} U\bar{\rho} U^\dagger] &=& \frac{1}{2^k}{\rm Tr}[\Pi_{\bm s} U \Pi_{\bm 0} U^\dagger] + \sum_{\bar{Q}\in \{\bar{I}, \bar{X}, \bar{Y}, \bar{Z}\}^{\otimes k}\backslash \{\bar{I}\}} c_{\bar{Q}} {\rm Tr}[\Pi_{\bm s} U \Pi_{\bm 0} \bar{Q} U^\dagger]\\
    &=& \frac{1}{2^k}{\rm Tr}[\Pi_{\bm s} U \Pi_{\bm 0} U^\dagger],
\end{eqnarray}
which is indeed independent of the logical state $\bar{\rho}$.
Therefore, by \Cref{lem:fact}, $U$ yields a weak transversal operation.

In order to complete the proof of \Cref{thm:z-rotation}, we must assure that any logical rotation angle $\bar{\theta}$ can be implemented, as follows.

\begin{lemma} \label{lem:logical-angle-zonly}
Assume that a CSS code satisfies the conditions given in \Cref{thm:z-rotation}, and let
    $U$ be a unitary defined as $U = \prod_{j \in \supp(\bar{Z}_{\rm target})} e^{i \theta Z_j}$. Then, by applying $U$ on an arbitrary logical state $\bar{\rho}$, followed by syndrome measurement $\mathcal{M}$, the gate operation can be written as
    \begin{eqnarray}
        \mathcal{M} \circ \mathcal{U}(\bar{\rho}) = \sum_{\substack{\bm s \in \{0,1\}^{M-1}}} \mathcal{E}_{\bm s} (p_{\bm s} \bar{V}_{\bm s} \bar{\rho} \bar{V}_{\bm s}^\dagger),        \label{eq:weak-transversal-Z-expression}
    \end{eqnarray}
    where $\mathcal{E}_{\bm s}(\cdot)= E_{\bm s} \cdot E_{\bm s}^\dagger$ is a map for $E_{\bm s}$, a minimum-weight Pauli error that invokes syndrome $\bm s$, $p_{\bm s}$ is the probability of obtaining the syndrome $\bm s$, and $\bar{V}_{\bm s} = e^{i \bar{\theta}_{\bm s} \bar{Z}_{\rm target}}$ is a logical rotation unitary.
    The probabilities and angles are given by
    \begin{align}
        p_{\bm s} &=
        |\cos^{d_z - w} \theta \sin^{w}\theta|^2
        +
        |\sin^{d_z - w} \theta \cos^{w}\theta|^2
        ,  \nonumber \\
        \bar{\theta}_{\bm s} &= \arctan\left(
        (-1)^{\frac{d_z-1}{2} - w} \tan^{d_z - 2w}\theta
        \right), \label{eq:logical_Z_prob_angles}
    \end{align}
    where $w := |E_{\bm s}|$ indicates the Pauli weight of $E_{\bm s}$.
\end{lemma}

\begin{proof}

There are $d_z - 1$ X-type stabilizer generators that can be related with
individual entry of a syndrome ${\bm s} = (s_1, ..., s_{d_z-1})$.
Since the Z code distance of the target QECC is $d_z$,  any Pauli Z error $E$ with ${\rm  wt}(E)\leq \frac{d_z-1}{2}$ must be discriminated by the syndrome measurement.
Considering that the number of combinations of such Pauli errors is
$$\sum_{l=0}^{\frac{d_z-1}{2}}\begin{pmatrix}
    d_z \\
    l
\end{pmatrix} = \frac{1}{2}\sum_{l=0}^{d_z}\begin{pmatrix}
    d_z \\
    l
\end{pmatrix} = 2^{d_z-1},$$
we can see that there exists a set of stabilizer generators $\{g_i\}_{i=1}^{d_z-1}$ for correcting the error on $\supp_Z(\bar{Z}_{\rm target})$.
In other words, we can uniquely relate a syndrome $\bm s$ with a Pauli Z error $E$ with ${\rm wt}(E)\leq \frac{d_z-1}{2}$, which is denoted as $E_{\bm s}$ in the following.

We remark that $E_{\bm s}$  can be estimated with an optimal decoding algorithm for a {\it classical} code with $d_z$ bits. Namely, instead of considering the stabilizer generator $g = \prod_{i=1}^{n} Q_i~(Q_i\in \{I, X, Y, Z\})$, one may define a punctured stabilizer generator as
$$
\hat{g}:= \prod_{i \in \supp_Z(\bar{Z}_{\rm target})} Q_i,
$$ and construct a classical code whose check matrices are given by $H = \begin{pmatrix}
    \rule[0.4ex]{1em}{0.4pt}~
    \hat{a}_i^T~\rule[0.4ex]{1em}{0.4pt}
\end{pmatrix}^T$ where the binary vectors $\hat{a}_i$ are defined from $\hat{g}_i = X^{\hat{a}_i}.$
Thus, the computational cost for computing $E_{\bm s}$ is drastically reduced compared to decoding the original QECC. For a moderate $d_z$, one may prepare a lookup table for it.

Now we decompose the unitary as
\begin{align}
    U &= \prod_{j \in \supp_Z(\bar{Z}_{\rm target})} e^{i \theta Z_j} = \prod_{j \in \supp_Z(\bar{Z}_{\rm target})}\left( \cos \theta \cdot I + i \sin \theta \cdot Z_j\right) \\
    &= \sum_{{\bm b}\in \{0, 1\}^{d_z}} (\cos \theta)^{d_z - |\bm b|} (i \sin \theta)^{|\bm b|} Z^{{\bm b}} \\
    &= \sum_{\substack{{\bm b}\in \{0, 1\}^{d_z}\\|{\bm b}|\leq \frac{d_z-1}{2}}} (\cos \theta)^{d_z - |{\bm b}|} (i \sin \theta)^{|{\bm b}|} Z^{{\bm b}}  + (\cos \theta)^{|\bm b|} (i \sin \theta)^{d_z - |\bm b|} Z^{\bar{\bm \beta} + {\bm b}}, \\
\end{align}
where $\bar{Z}_{\rm target} = Z^{\bm \beta}.$ As we have noted above, $Z^{\bm b}$ with $|\bm b| \leq \frac{d_z - 1}{2}$ can be uniquely identified with a syndrome ${\bm s}$. Thus, we can further rewrite the unitary using projector $\Pi_{\bm s} := E_{\bm s} \Pi_{\bm 0} E_{\bm s}$ and $w \coloneqq |E_{\bm s}|$  as
\begin{align}
    U &= \sum_{\bm s} \Pi_{\bm s} U = \sum_{\bm s} U_{\bm s}, \\
    U_{\bm s} &= E_{\bm s} \cdot (\cos \theta)^{{\rm wt}(E_{\bm s})} (i \sin\theta)^{{\rm wt}(E_{\bm s})} (\cos^{d_z - 2 {\rm wt}(E_{\bm s})} \theta\cdot I + (i \sin\theta)^{d_z - 2 {\rm wt}(E_{\bm s})}  \cdot \bar{Z}_{\rm target}) \nonumber \\
    &\simeq  \sqrt{p_{\bm s}} E_{\bm s} \bar{V}_{\bm s},
\end{align}
where we have neglected the global phase in the second line, and $p_{\bm s}$ and $\bar{V}_{\bm s}$ are defined as in Eq.~\eqref{eq:logical_Z_prob_angles}.

Finally, by using the fact that  $\mathcal{M}(\Pi_{\bm s} \cdot \Pi_{{\bm s}'}) = \delta_{{\bm s}, {\bm s}'} \mathcal{M}(\Pi_{\bm s}\cdot \Pi_{\bm s})$ holds for the syndrome measurement $\mathcal{M}$ of the QECC, we obtain
\begin{align}
    \mathcal{M} \circ \mathcal{U}(\bar{\rho}) &= \mathcal{M}(U \bar{\rho} U^\dagger) = \sum_{{\bm s}, {\bm s}'} \mathcal{M}(\Pi_{\bm s} U \bar{\rho} U^\dagger \Pi_{{\bm s}'}) = \sum_{\bm s} \mathcal{M}(\Pi_{\bm s} U \bar{\rho} U^\dagger \Pi_{\bm s}) \nonumber \\ &= \sum_{\bm s} \mathcal{M} \circ \mathcal{E}_{\bm s} (p_{\bm s}\bar{V}_{\bm s} \bar{\rho} \bar{V}_{\bm s}^\dagger) \nonumber \\
    &= \sum_{\bm s}\mathcal{E}_{\bm s} (p_{\bm s} \bar{V}_{\bm s} \bar{\rho} \bar{V}_{\bm s}^\dagger),
\end{align}
where we have used $\mathcal{M} \circ \mathcal{E}_{\bm s} = \mathcal{E}_{\bm s} \circ \mathcal{M}$ and $\mathcal{M}(\bar{V}_{\bm s} \bar{\rho} \bar{V}_{\bm s}^\dagger) = \bar{V}_{\bm s} \bar{\rho} \bar{V}_{\bm s}^\dagger$ in the last line.
\end{proof}

Note that $\mathcal{M}$ in \Cref{lem:logical-angle-zonly} can be in principle replaced by syndrome measurement for the punctured stabilizer generators $\{\hat{g}_i\}_i$ if we are only interested in the weak transversal gate. However,  in practice we desire to suppress the hardware error on the physical qubits, and hence perform the syndrome measurement for the entire QECC.
By combining \Cref{lem:logical-angle-zonly} with the fact that $U$ is a weak transversal gate, we complete the proof of \Cref{thm:z-rotation}.

We note that Theorem~\ref{thm:z-rotation} also implies weak transversal X rotations after exchanging X and Z. If the assumptions in Theorem~\ref{thm:z-rotation} are satisfied for both X and Z operators, a weak transversal Y rotation can also be obtained by a constant-size composition of weak transversal X and Z rotations, for example using $e^{i\theta Y}=e^{-i\pi Z/4}e^{i\theta X}e^{i\pi Z/4}$. This composition uses multiple syndrome-extraction gadgets. In contrast, the disjoint-support construction in Appendix~\ref{app:thm2-proof} gives a single-gadget implementation of general logical Pauli rotations when its conditions are satisfied.

\section{Proof of Theorem 2}\label{app:thm2-proof}

In this Appendix, we provide a detailed proof for \Cref{thm:general-pauli-rot}, which we restate here for convenience.
\weakPthm*

The proof proceeds as follows.
In Appendix~\ref{subsec:disjoint} we claim that, once a partition of logical Pauli operator $\bar{P}$ called {\it disjoint support partition} is obtained, we can implement $\bar{P}$-rotation gate in weak transversal manner for arbitrary logical angle.
Motivated by actually constructing the disjoint support partition, in Appendix~\ref{subsubsec:multi-XZ} we provide Algorithm~\ref{alg:pivoting_xz} for finding the disjoint support partition.

\subsection{Disjoint support partition for weak transversal Pauli rotation} \label{subsec:disjoint}

In order to ensure the input-agnosticity, it is important that we perform physical Pauli rotations with appropriate support.
Here, we consider a specific partition of logical operator as follows.

\begin{definition} (Disjoint support partition)\label{def:disjoint-support-partition}
Let $\bar{P}$ be a nontrivial Pauli operator of an $[[n, k, d]]$ CSS code.
A set of Pauli operators $\{P_m \}_{m=1}^M~(P_m\in \{I, X, Y, Z\}^{\otimes n})$ is said to be a disjoint support partition of $\bar{P}$ if the following conditions are satisfied:
\begin{itemize}
    \item \black{$\bar{P} = \prod_{m=1}^M P_m$ with odd $M$.}
    \item $\supp(P_m) \cap \supp(P_{m'}) = \emptyset$ for all $m \neq m'$.
    \item For any nonempty subset $R \subsetneq [M]$,  it holds that $\Pi_{\bm 0} \prod_{m \in R}P_m \Pi_{\bm 0} = 0$ where $\Pi_{\bm 0}$ is a projection operator onto the code space.
\end{itemize}
\black{In particular, the third condition assures that any nonempty proper subset of the operators yield nontrivial syndrome.}
\end{definition}

Assuming that $\{P_m\}_m$ is a disjoint support partition, we can show that $U=\prod_{m=1}^M e^{i \theta P_m}$ is a weak transversal Pauli rotation gate.
\begin{lemma} \label{lem:weak-pauli-rot-disjoint}
    Let $U=\prod_{m=1}^M e^{i \theta P_m}$ with $\{P_{m}\}_{m=1}^M$ being a disjoint support partition of the target Pauli operator $\bar{P}_{\rm target}$.
    Then, $U$ followed by syndrome measurement is a weak transversal $\bar{P}_{\rm target}$ rotation gate.
\end{lemma}
\begin{proof}
    First, let us decompose the unitary into sum of Pauli operators as
    \begin{eqnarray}
        U = \prod_{m=1}^M e^{i \theta P_m} = \sum_{{\bm b}\in \{0, 1\}^{M}} \alpha_{\bm b} \prod_{m=1}^M P_m^{b_m}
    \end{eqnarray}
    where $\alpha_{\bm b} \in \mathbb{C}$ is a coefficient of product of Pauli operators.
Since $\{P_m\}_m$ is a disjoint support partition, for a given proper subset $R \subsetneq [M]$, only $R^c = [M]\setminus R \neq R$ is a proper subset that satisfies  $(\prod_{m\in R} P_m)\cdot (\prod_{m \in R^c}P_m) = \bar{P}_{\rm target} \in \{\bar{I}, \bar{X}, \bar{Y}, \bar{Z}\}^{\otimes k}$.
Therefore,
\begin{align}
\Pi_0 (\prod_{m\in R} P_m) \cdot (\prod_{m\in R'} P_m) \Pi_0 = \begin{cases}
    \Pi_0 & \text{if $R' = R$}\\
    \Pi_0 \bar{P}_{\rm target} & \text{if $R' = R^c$} \\
    0 & \text{otherwise}.
\end{cases}\label{eq:P_product_condition}
\end{align}
Note that $\prod_m P_m^{\bm b_m}$ can be uniquely related to a correction operation $E_{\bm s(\bm b)}$ where $\bm s(\bm b)$ indicates the syndrome invoked by the Pauli operator. Combining $\bm s(\bm b) = \bm s(\bm 1- \bm b)$ with Eq.~\eqref{eq:P_product_condition}, it follows that
\begin{eqnarray}
    \Pi_{\bm 0} E_{\bm s} E_{{\bm s}'} \Pi_{\bm 0} = \delta_{{\bm s}, {\bm s}'} \Pi_{\bm 0}, \label{eq:KL-condition}
\end{eqnarray}
which satisfies the Knill-Laflamme error correcting condition~\cite{knill1997theory}.
This indicates that we can discriminate the elements of $\{E_{\bm s}\}$ from each other via syndrome measurement.

Observe that, since $M$ is odd due to the condition of disjoint support partition, we can decompose the unitary as
\begin{align}
U = \sum_{\bm s} U_{\bm s}\label{eq:unitary_decomposition_by_s}
\end{align}
where
\begin{align}
    U_{\bm s}
    &= E_{\bm s}(\cos \theta)^{w} (i \sin \theta)^{w} (\cos^{M - 2w}\theta + (i \sin \theta)^{M - 2w} \bar{P}_{\rm target}) \\
    &\simeq \sqrt{p_{\bm s}} E_{\bm s} \bar{V}_{\bm s},
\end{align}
where $w \coloneqq |E_{\bm s}|$. This is not even a unitary operator if $M$ is an even number.
Note that we have neglected the global phase in the second line.
Here, $p_{\bm s} \in \mathbb{R}$ is the normalization factor, and $\bar{V}_{\bm s} = e^{i \bar{\theta}_{\bm s} \bar{P}_{\rm target}}$ is the  logical rotation unitary that are explicitly given as
\begin{align}
    p_{\bm s} &= |\cos^{M - w} \theta \sin^{w}\theta|^2 + |\sin^{M - w}\theta \cos^{w}\theta|^2, \\
    \bar{\theta}_{\bm s} &= \arctan\left((-1)^{\frac{M-1}{2} - w} \tan^{M - 2w} \theta\right).
\end{align}

Now, let us confirm that we indeed obtain $\bar{V}_{\bm s}$ with probability $\bm s$ when the syndrome measurement result is $\bm s$.
Let $\Pi_{\bm s} = E_{\bm s} \Pi_0 E_{\bm s}$ be the Kraus operator corresponding to the measurement with syndrome $\bm s$. Due to Eq.~\eqref{eq:KL-condition}, it holds that $\Pi_{\bm s} U \Pi_0 = \Pi_{\bm s} U_{\bm s} \Pi_0$.
Then, it follows that, for any logical operator $\bar{Q}\in\{\bar{I}, \bar{X}, \bar{Y}, \bar{Z}\}^{\otimes k}$,
\begin{align}
    {\rm Tr}[\Pi_{\bm s} U \Pi_0 \bar{Q} U^\dagger] &=
    {\rm Tr}[\Pi_{\bm s} U_{\bm s} \Pi_0 \bar{Q} U_{\bm s}^\dagger] \nonumber \\
    &= p_{\bm s} {\rm Tr}[\Pi_{\bm s}E_{\bm s} \bar{V}_{\bm s}\Pi_0 \bar{Q} E_{\bm s}\bar{V}_{\bm s}^\dagger] \nonumber\\
    &= p_{\bm s}{\rm Tr}[\Pi_0 \bar{Q}]\nonumber \\
    &=\begin{cases}
        p_{\bm s} \cdot 2^k & \text{if $\bar{Q}=\bar{I}$} \\
        0 & \text{otherwise}
    \end{cases}. \label{eq:trace-eval}
\end{align}
From the second to third line, we have used the fact that operators consisting of logical operators ($\bar{Q}$ and $\bar{V}_{\bm s}$) and those consisting only of stabilizers and/or destabilizers ($E_{\bm s}$ and $\Pi_{\bm s}$) commute with each other.
Now decompose a logical state as
\begin{eqnarray}
    \bar{\rho} = \Pi_{\bm 0} \left(\frac{I}{2^k} + \sum_{\bar{Q}\in \{\bar{I}, \bar{X}, \bar{Y}, \bar{Z}\}^{\otimes k}} c_{\bar{Q}} \bar{Q}\right),\label{eq:logical_state}
\end{eqnarray}
where $c_{\bar{Q}}$ is the coefficient that encodes the logical information of the state.
By combining Eqs.~\eqref{eq:trace-eval},~\eqref{eq:logical_state} we find that the probability of obtaining the syndrome $s$ is independent of the logical state, since
$
{\rm Tr}[\Pi_{\bm s} U \bar{\rho}U^\dagger] = p_{\bm s}.
$
Concretely, by denoting syndrome measurement channel as $\mathcal{M}$ and the unitary channel as $\mathcal{U}$,  we find that
\begin{align}
    \mathcal{M} \circ \mathcal{U}(\bar{\rho}) =
    \mathcal{M}(U \bar{\rho} U^\dagger)
    &= \sum_{\bm s} p_{\bm s} \mathcal{E}_{\bm s} (\bar{V}_{\bm s} \bar{\rho} \bar{V}_{\bm s}^\dagger),
\end{align}
where we have defined $\mathcal{E}_{\bm s}(\cdot) \coloneqq E_{\bm s}\cdot E_{\bm s}$.
This indicates that unitary $U$ followed by syndrome measurement results in implementing logical unitary $\bar{V}_{\bm s}$ with probability $p_{\bm s}$, associated with Pauli error $E_{\bm s}$.
\end{proof}

\subsection{Construction of disjoint support partition}
\label{subsubsec:multi-XZ}
The argument in the previous section motivates us to construct an algorithm to compute disjoint support partition. Here, we provide an efficient algorithm to determine the disjoint support partition of arbitrary logical operator into partition number of an integer $M \le d$.

Let $H_X$ and $H_Z$ be the X- and Z-check matrices of an $[[n, k, d]]$ code, and let a syndrome vector for any Pauli operator $P=X^{\bm a} Z^{\bm b}$
be defined as
\begin{equation}
\bm s(\bm a, \bm b) \coloneqq (H_Z \bm a^\top, H_X \bm b^\top).
\end{equation}
Let $\bar P_{\mathrm{target}}=X^{\bm \alpha} Z^{\bm \beta}$ be a physical Pauli operator representation of a nontrivial logical Pauli operator, and let
\begin{equation}
S := \supp(\bm \alpha)\cup\supp(\bm \beta) \subseteq [n]
\end{equation}
be its support.  For each $j\in S$, we define its syndrome column vector in $\mathbb{F}_2^{r_Z+r_X}$ by
\begin{align}
\bm h_j &:= \bm s(\alpha_j \bm e_j, \beta_j \bm e_j) =
\bigl(\alpha_j\,H_Z(:,j),\; \beta_j\,H_X(:,j)\bigr),\\
\label{eq:hj}
&= \begin{cases}
\binom{H_Z \bm e_j}{0} \in \mathbb{F}_2^{r_Z+r_X} & \text{if } (\alpha_j, \beta_j)=(1, 0) \ (\text{syndrome vector for }X_j),\\[1ex]
\binom{0}{H_X \bm e_j} \in \mathbb{F}_2^{r_Z+r_X} & \text{if } (\alpha_j,\beta_j)=(0, 1) \ (\text{syndrome vector for }Z_j),\\[1ex]
\binom{H_Z \bm e_j}{H_X \bm e_j} \in \mathbb{F}_2^{r_Z+r_X} & \text{if } (\alpha_j, \beta_j)=(1, 1) \ (\text{syndrome vector for }Y_j),\\[1ex]
\end{cases}
\end{align}
where $H(:,j)$ denotes the $j$-th column and $\bm e_j$ is the $j$-th standard basis vector.
Note that $\bm s(\bm \alpha,\bm \beta)=\sum_{j\in S} \bm h_j$.

As discussed in the previous section, if $\{P_m\}_{m=1}^{M}$ is a disjoint support partition of the logical operator $\bar{P}_{\rm target}$, for any proper subset $ R\subsetneq [M]$, the product $\prod_{m\in R} P_m$ invokes a different syndrome. This can be formally stated as the following Lemma.

\begin{lemma}[Rank of restricted check matrix]\label{lem:rank_condition_xz}
Let $C$ be an $[[n,k,d]]$ CSS code with check matrices
$H_X\in\mathbb{F}_2^{r_X\times n}$ and $H_Z\in\mathbb{F}_2^{r_Z\times n}$.
Let $\bar{P}_{\mathrm{target}}=X^{\bm \alpha}Z^{\bm \beta}$ be a nontrivial logical Pauli operator, and define its support as $S \coloneqq \supp(\bm \alpha)\cup \supp(\bm \beta)$.
Let $H_S$ be the $(r_Z+r_X)\times |S|$ matrix whose columns are $\{\bm h_j\}_{j\in S}$.
Then,
the rank of the restricted check matrix can be bounded from below as
\begin{equation}
{\rm rank}(H_S)\ \ge\ d-1. \label{eq:rank-condition}
\end{equation}
\end{lemma}

\begin{proof}
Because $\bar{P}_{\mathrm{target}}$ is a logical operator, it commutes with all stabilizer
generators, hence it has trivial syndrome. For CSS codes, this means
$H_Z \bm \alpha^\top=\bm 0$ and $H_X \bm \beta^\top=\bm 0$. Equivalently,
\[
\sum_{j\in S} \bm h_j \;=\; 0,
\]
implying that the column set $\{\bm h_j\}_{j\in S}$ is linearly dependent.

We show Eq.~\eqref{eq:rank-condition} by contradiction.
Assume for contradiction that $\rank(H_S)\le d-2$.
Since the columns are dependent, there exists a minimal dependent subset
$D\subseteq S$ that satisfies $\sum_{j\in D} \bm h_j=0$ and every nonempty proper subset of $\{\bm h_j\}_{j \in D}$
is linearly independent. For such a $D$, we have $\rank(H_D)=|D|-1$ where $H_D$ is the restriction
of $H_S$ to the columns indexed by $D$. Therefore,
\[
|D|-1=\rank(H_D)\le \rank(H_S)\le d-2 \quad\Rightarrow\quad |D|\le d-1.
\]

Let $\bar{P}_{\rm target}|_D$ be the restriction of $\bar{P}_{\mathrm{target}}$ to the qubits in $D$.
By construction, $\sum_{j\in D} \bm h_j=0$ holds, hence $P_D$ has trivial syndrome and thus lies in
the normalizer. Its weight is $\mathrm{wt}(P_D)=|D|\le d-1<d$.
By the definition of the code distance $d$, no nontrivial logical operator can have weight
$<d$. Hence $P_D$ must be a stabilizer (up to a phase).

Now multiply $\overline{P}_{\mathrm{target}}$ by this stabilizer $P_D$; this removes the
factors on $D$ while keeping the syndrome trivial. Repeating the same argument on the
remaining support, we can decompose $\overline{P}_{\mathrm{target}}$ as a product of
stabilizers, implying that it itself is a stabilizer, contradicting the nontriviality.
Therefore, $\rank(H_S)\ge d-1$.
\end{proof}

Using this Lemma, we can prove the following Proposition.

\begin{proposition}[Construction of disjoint support partition]
\label{prop:multi-xz}
Let $C$ be an $[[n,k,d]]$ CSS code, and $M$ be an odd number satisfying $M \leq d$.
Then, Algorithm~\ref{alg:pivoting_xz} outputs binary vectors $\{(\bm A_m,\bm B_m)\}_{m=1}^M$ such that
$\{X^{\bm A_m}Z^{\bm B_m}\}_{m=1}^M$ is a disjoint support partition of $\bar{P}_{\mathrm{target}}$ in the sense of Definition~\ref{def:disjoint-support-partition}.
\end{proposition}

\begin{proof}
Algorithm~\ref{alg:pivoting_xz} chooses $S' \subsetneq S$ with $|S'| = M-1$ such that $\{\bm h_j\}_{j\in S'}$ are linearly independent,
which exists by Lemma~\ref{lem:rank_condition_xz}. Define
\[
P_m := X^{\bm A_m}Z^{\bm B_m}.
\]
Because any $P_m$ with $1 \leq m \leq M-1$ act on distinct single qubits and $P_M$ acts on the remaining qubits,
the supports are pairwise disjoint and $\bar {Z}_{\mathrm{target}}=\prod_{m=1}^M P_m$.

For the condition where there is no proper subproduct on code space,
note that $\bm s(\bar P_{\mathrm{target}})=0$ and syndromes that add up under multiplication of operators.
Due to the independence of syndromes, for every nonempty proper subset $R \subsetneq [M]$, it also holds that $\bm s(\prod_{m\in R} P_m) \neq 0$, which implies $\Pi_0 (\prod_{m\in R} P_m)\Pi_0 = 0$ for every such $R$. Together  with the above argument, $\{P_m\}_{m=1}^M$ is a disjoint support partition of $\bar{P}_{\rm target}$.

\end{proof}

\begin{algorithm}
\caption{Disjoint support partition for a general logical Pauli operator}
\label{alg:pivoting_xz}
\KwIn{CSS check matrices $H_X\in \mathbb{F}_2^{r_X\times n}$, $H_Z\in \mathbb{F}_2^{r_Z\times n}$; \\
\hspace{3.7em}
Target logical operator
$P_{\mathrm{target}}= \bar{X}^{\bm \mu} \bar{Z}^{\bm \nu} = X^{\bm \alpha}Z^{\bm \beta}$.
}
\KwOut{Disjoint support partition $\{X^{\bm A_m}Z^{\bm B_m}\}_{m=1}^M$ for $P_{\mathrm{target}}$.
}

$S \gets \supp(\bm \alpha)\cup\supp(\bm \beta)$\;

\For{$j\in S$}{
    $\bm h_j \gets (\alpha_j\,H_Z(:,j),\; \beta_j\,H_X(:,j))
    \in \mathbb{F}_2^{r_Z+r_X}$\;
}

Find $S'\subsetneq S$ such that $|S'| = M-1$ and
$\{\bm h_j\}_{j\in S'}$ are linearly independent, e.g., via Gaussian elimination\;

\For{$j_m \in S'$}{
$X^{\bm A_m} Z^{\bm B_m} \gets \bar{P}_{\rm target}|_{j_m}$\;
}

$\bm A_M \gets \bm \alpha - \sum_{m=1}^{M-1}\bm A_m, \quad \bm B_M \gets \bm \beta - \sum_{m=1}^{M-1}\bm B_m$\;

\Return{$\{(\bm A_m,\bm B_m)\}_{m=1}^M$}\;
\end{algorithm}

With Lemma~\ref{lem:weak-pauli-rot-disjoint} and \Cref{prop:multi-xz}, we complete the proof of \Cref{thm:general-pauli-rot}.

\section{Coherent error suppression in probabilistic implementation} \label{app:mixed-synthesis}

Let $\mathcal{H}$ be the target Hilbert space of interest, and $\mathcal{L}(\mathcal{H})$ and $\mathcal{B}(\mathcal{H})$ be the set of linear operator and bounded linear operators acting on $\mathcal{H}$.
Let the diamond distance between two quantum maps $\Phi : \mathcal{L}(\mathcal{H}) \to \mathcal{L}(\mathcal{H})$ and $\Psi:\mathcal{L}(\mathcal{H}) \to \mathcal{L}(\mathcal{H})$ be defined as follows,
\begin{align}
    \ddist(\Phi, \Psi) := \| \Phi - \Psi\|_{\diamond}/2,
\end{align}
where the diamond norm for a map is defined using the trace norm $\|\cdot\|_1 := \tr[\sqrt{\cdot \cdot ^\dagger}]$ as
\begin{align}
    \|\Phi\|_{\diamond} = \max_{X\in \mathcal{B}(\mathcal{H})} \frac{\|(\Phi\otimes \mathds{1}) (X)\|_1}{\|X\|_1}.
\end{align}

Now let us consider a Pauli Z rotation channel $\mathR_\theta = e^{i \theta Z} \cdot e^{-i\theta Z}.$
It can be shown that the diamond distance from the identity channel is
\begin{align}
    \ddist(\mathR_\theta, \mathI) = \ddist(\mathR_{-\theta}, \mathI) = |\sin \theta|.
\end{align}
Next, assume that we cannot implement the identity channel itself (or the noiseless target unitary channel), but we have access to channels that deviate by $\mathR_{\theta}$ and $\mathR_{-\theta}$. Then, by simply taking the average over the channels, we can quadratically suppress the error~\cite{hastings2016turning}:
\begin{align}
\ddist\left(\frac{\mathR_{\theta} + \mathR_{-\theta}}{2}, \mathI\right) = \sin^2\theta.
\end{align}

Furthermore, assume that the angles of the accessible channels $\mathR_{\theta_1>0}$ and $\mathR_{\theta_2<0}$ do not sum up to zero, i.e., $\theta_1 + \theta_2 \neq 0$. Even in such a case, it can be shown that optimal mixture of two channels achieve quadratic improvement, as follows~\cite{akibue2024probabilistic, yoshioka2025error}:
\begin{lemma}
    Let $\mathR_{\theta_1}$ and $\mathR_{\theta_2}$ be Z-rotation channels with $\theta_1 > 0$ and $\theta_2 < 0,$ and let $p_1, p_2 \in [0, 1]$ satisfy $p_1+p_2=1$.
    Then, it holds that
    \begin{align}
        \min_{p_1,p_2}\ddist\left(p_1 \mathR_{\theta_1} + p_2 \mathR_{\theta_2}, \mathI\right) \leq \epsilon^2,
    \end{align}
where $\epsilon = \max\{\ddist(\mathR_{\theta_1}, \mathI), \ddist(\mathR_{\theta_2}, \mathI)\}$.

\end{lemma}

\section{Details on numerics for repeat-until-success protocol}\label{app:rus-numerics}

\subsection{Decoders for repeat-until-success}
\label{subsubsec:inplace-rot-decoders}
In the circuit-level repeat-until-success (RUS) simulations, the syndrome
record contains two qualitatively different contributions.  
The first is the
usual signal caused by physical faults during syndrome
extraction, and the second is induced by weak transversal gadget. When only a single round of RUS is performed, as in the small angle regimes or postselected case, we numerically found that it suffices to treat with an {\it one-step decoder} which does not discriminate two errors. Meanwhile, in the case when we perform multiple rounds of RUS, the failure in earlier rounds applies physical Pauli operations to some physical qubits, and hence one must discriminate two errors so that memory decoder would not harmfully introduce additional error.
We refer to the RUS-aware decoder that separates these two
contributions before applying the ordinary memory decoder as the {\it two-step
decoder},
and describe how it is implemented in the following.

Let $s_{\rm raw}^{(r)}$ be the raw syndrome obtained from the $r$-th round of RUS branch, and the syndromes stored in the memory be given as $s^{(r-1)}_{\rm mem}$ and $s_{\rm RUS}^{(r-1)}$ for the memory and RUS protocol, respectively. 
We first construct a detector $\Delta s^{(r)} = s_{\rm raw}^{(r)} \oplus s^{(r-1)}_{\rm raw}$. Then, one use a pre-computed lookup table that considers the entire neighboring check detector pattern to decide whether this increment is attributed to error supported on the weak transversal gadget or not. This enable us to split the detector as $d_{\rm RUS}^{(r)}\oplus d_{\rm mem}^{(r)} = \Delta s^{(r)}$.
The RUS detector $d_{\rm RUS}^{(r)}$ determines the physical operation $E_{\bm s}$ applied to the logical support, and hence related with the logical angle applied to the system. This is used to decide whether or not to terminate the RUS protocol.
The memory detector $d_{\rm mem}^{(r)}$ is used to compute the memory syndrome $s_{\rm mem}^{(r)} = s_{\rm mem}^{(r-1)}\oplus d_{\rm mem}^{(r)}$, which is stored on a classical register. 
If the RUS protocol finishes in $R$ rounds, one can perform decoding for the entire set of memory detector $\{d_{\rm mem}^{(r)}\}_{r=1}^R.$

For the numerical simulation of weak transversal gadget, we have initialized the system so that no nontrivial syndrome exists at the beginning of the RUS protocol, as $s_{\rm raw}^{(0)} = s_{\rm mem}^{(0)} = s_{\rm RUS}^{(0)} = 0$. Also, after the RUS protocol, we have performed transversal Pauli measurement on the data qubits, and thus the decoding graph for the memory is similar to the ordinary memory experiment. We leave the extension to window decoding as an interesting future work.

\subsection{Optimization of implementation policy}

The RUS protocol for weak transversal gadget is not unique, in the sense that we have degrees of freedom to choose physical angles and termination condition. Here, we discuss how to optimize the protocol in two stages.
First, we use inexpensive repetition-code simulations, and, when appropriate, code-capacity branch-policy searches, to identify promising termination policies and physical-angle rules.
Second, we evaluate the resulting candidates in the circuit-level surface-code simulation (using matrix product state, or MPS) and selected the best point after scanning the maximum number of RUS rounds.
This separation is useful because the latter is too expensive to optimize directly over both the branch policy and the continuous physical angles.

For the one-step decoder, we optimized the physical angles by minimizing the trace distance
between the output logical state $\bar{\rho}_{\rm out}$ and the ideal state $\bar{\rho}_{\rm ideal}$ in the repetition-code density-matrix simulation:
\begin{align}
    D_{\rm tr}(\rho_{\rm out},\rho_{\rm ideal})
    := \frac{1}{2}\|\rho_{\rm out}-\rho_{\rm ideal}\|_1.
\end{align}
The one-step decoder does not apply a Pauli-frame correction before the final maximum-likelihood decoding, so the dominant effect of the angle optimization is to make the averaged logical output channel close to the target logical rotation.
As briefly mentioned in the privious subsection, for the small-angle regime considered in the surface-code simulations, the best repetition-code strategy was obtained with a single RUS round, $R=1$, using the optimized physical angle.
As a representative example, for $\barthetatarget=10^{-4}$, the unoptimized physical angle makes the trivial-syndrome branch close to the target angle, but the rare nontrivial-syndrome branches produce a large opposite rotation.
The optimized physical angle shifts the trivial branch slightly away from the target so that the average logical channel is much closer to the desired unitary.
The same optimized $R=1$ one-step strategy was then used in the circuit-level simulation for surface code, where it remained the best low-angle candidate among the strategies tested.

For the two-step decoder, we considered a broader class of policies because the Pauli-frame correction makes the result sensitive to the history of observed syndromes.
We first tested fixed termination schedules, specified by a target value of the Pauli correction weight $\chi$ in each round. For instance, we denote $\chi_{\rm target}=[0,0, 1]$ to indicate that the RUS protocol terminates if $\chi = 0$ for the first and second round, and $\chi=1$ for larger number of rounds.
For these fixed schedules, we compared analytic, unoptimized physical angles against physical angles optimized in the density-matrix simulation for repetition code.
While we did not exhaustively explore in this work,  the policy family can be even more diverse; we may perform a search over history-dependent branch policies based on the code-capacity noise  model.
Namely, the target termination value of $\chi$ can be chosen from the current branch history: the first round uses a root target, the next two failed rounds use maps depending on the last recorded $\chi$, and later failed rounds use a tail map. It is interesting direction how to efficiently optimize over such policies, and to investigate what is the potential room of further improvement.

\subsection{Error metric}

In the code-capacity simulations in the main text, we evaluate the logical channel itself and report the diamond distance $\ddist$ from the target logical unitary.
This channel-level metric is natural for comparing coherent logical errors, but it becomes numerically expensive in the circuit-level simulations.
In particular, reconstructing the full logical channel for a surface-code repeat-until-success circuit would require simulating the logical process for multiple input states while resolving the measurement-feedback branches.
We therefore use state-based error metrics for the circuit-level simulations.

For the circuit-level data, we initialize the logical qubit in $\ket{0_L}$ and compare the decoded output state $\rho_{\rm out}$ with the ideal output state $\rho_{\rm id}$ for the same target angle $\barthetatarget$.
When direct density-matrix simulation is available, we compute the trace distance between the output logical state $\bar{\rho}_{\rm out}$ and the ideal state $\bar{\rho}_{\rm ideal}$ as
\begin{align}
    D_{\rm tr}(\rho_{\rm out},\rho_{\rm ideal})
    := \frac{1}{2}\|\rho_{\rm out}-\rho_{\rm ideal}\|_1.
\end{align}
For circuit-level simulation for the surface code, we instead estimate an effective error from the logical Pauli expectation value.
For the logical $X$-rotation instances considered in the circuit-level simulations, the ideal output from $\ket{\bar{0}}$ satisfies
\begin{align}
    \langle \bar{Z}\rangle_{\rm ideal} = \cos(2\barthetatarget),
    \qquad
    \langle \bar{Y}\rangle_{\rm ideal} = \sin(2\barthetatarget),
\end{align}
up to the sign convention for the rotation direction.
We define
\begin{align}
    \epsilon_Z
    :=
    \frac{1}{2}
    \left|
    \langle \bar{Z}\rangle_{\rm out}
    -
    \langle \bar{Z}\rangle_{\rm ideal}
    \right|,
    \qquad
    \epsilon_Y
    :=
    \frac{1}{2}
    \left|
    \langle \bar{Y}\rangle_{\rm out}
    -
    \langle \bar{Y}\rangle_{\rm ideal}
    \right|.
\end{align}
In general, such expectation-value deviations are only lower bounds on the trace distance:
for any observable $O$ with $\|O\|_\infty\leq 1$, since
$
    \frac{1}{2}\left|\tr\left[O(\rho-\sigma)\right]\right|
    \leq
    D_{\rm tr}(\rho,\sigma)$.
In this sense, although $\epsilon_Z$ and $\epsilon_Y$ are not guaranteed to reproduce the full error for an arbitrary logical noise channel, we find that the proxy $\epsilon_Z$ is in particular faithful for simulations in this work.

\begin{figure*}[t]
    \centering
    \includegraphics[width=0.65\linewidth]{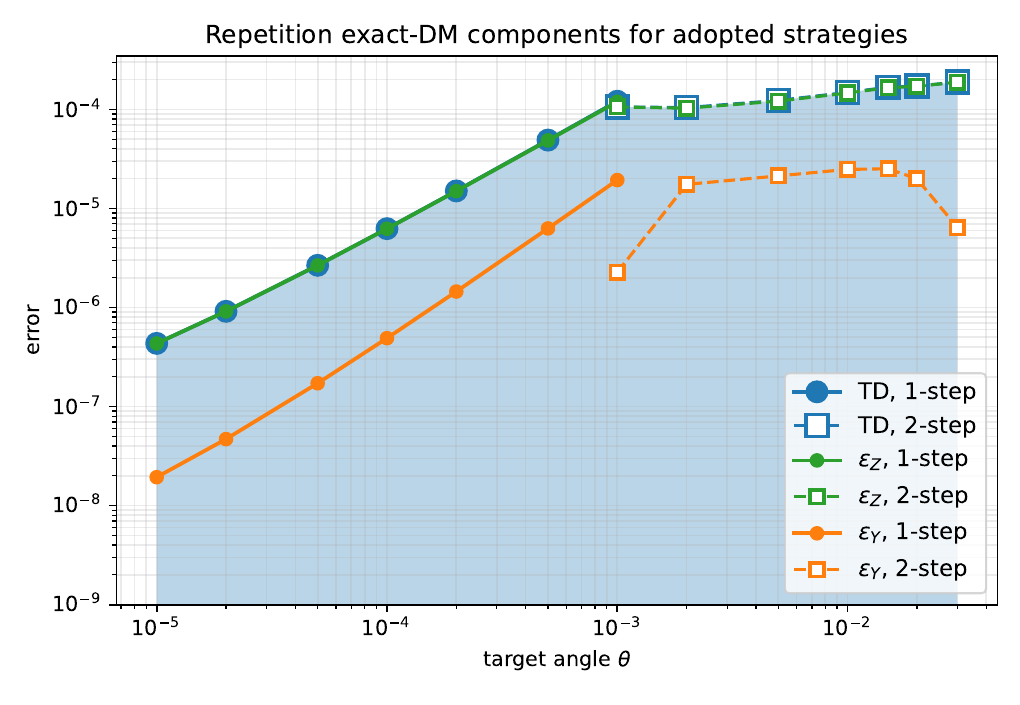}
    \caption{
    Validation of the circuit-level error metric using exact density-matrix simulations of the bit-flip repetition code.
    We compare the trace distance from the ideal logical output state with the effective errors inferred from logical Pauli expectation values, $\epsilon_Z$ and $\epsilon_Y$, for the same class of repeat-until-success strategies used in the surface-code simulations.
    In both the optimized one-step-decoder regime and the two-step-decoder regime, $\epsilon_Z$ closely tracks the trace distance, whereas $\epsilon_Y$ is subdominant.
    }
    \label{fig:app-error-metric-repetition-dm}
\end{figure*}

In Fig.~\ref{fig:app-error-metric-repetition-dm}, we show the numerical results of exact density-matrix circuit-level simulations for the bit-flip repetition code.
The $Z$-expectation-based estimate $\epsilon_Z$  reproduces the trace distance over the relevant range of target angles and for both decoder choices.
The $Y$-expectation contribution is much smaller in the optimized one-step-decoder regime and remains below the dominant $Z$ component in the two-step-decoder regime.
This indicates that the trace distance is dominated by the component visible in $\langle Z_L\rangle$ for the RUS strategies considered here.

The circuit-level simulation of the surface code does not give the full logical channel at comparable cost, but their error curves follow the same qualitative envelope as the circuit-level simulation of the bitflip code.
Moreover, the memory-only checks show that the independent memory contribution is suppressed below the RUS-induced error in the relevant parameter regime.
We therefore use $\epsilon_Z$ as the reported effective circuit-level error rate for surface code as well.

\subsection{Error scaling for small angles}
\label{subsec:rus-codecapacity-scaling}

Here, we briefly explain the error scaling of code-capacity simulation quoted in the main text.
In the small-angle regime, the ordinary RUS envelope is governed by first-order support faults that cause the decoder to assign a neighboring RUS branch.
The probability of such a branch-misclassification event is $\mathcal{O}(\pphys)$.
Expanding Eq.~\eqref{eq:logical-theta-expression} around the branch selected by the optimized policy and using the small-angle relation $\theta=\mathcal{O}(\bartheta^{1/M})$, the residual logical-angle error associated with one such misclassified branch scales as $\mathcal{O}(\bartheta^{1-1/M})$ for the code-capacity policy used in Fig.~\ref{fig:rus_history}(b).
The induced diamond distance therefore scales as
\begin{align}
    d_{\diamond}^{\rm RUS}
    =
    \mathcal{O}\!\left(\pphys \bartheta^{1-1/M}\right)
    + \mathcal{O}(\pphys^2),
\end{align}
up to constants depending on the selected branch policy.
This estimate is intended for the ordinary RUS envelope; the diluted protocol changes the angle dependence by probabilistically mixing this larger-angle channel with the identity channel, as discussed in Appendix~\ref{app:mixed-synthesis}.

\subsection{Implemented protocol}
\label{subsec:surface-m3-implemented-protocol}

For the rotated surface-code circuit-level simulations with $d=M=3$, we used the protocols summarized in \cref{tab:surface-m3-implemented-protocol}.
The reported logical error is the $Z$-based effective logical error estimated by the circuit level simulation of surface code using MPS.

\input{tables/surface_m3_implemented_protocol_table.tex}

\section{CNOT scheduling during syndrome extraction}\label{app:scheduling-d7-surface}
As described in the main text, a weak transversal Pauli rotation can be implemented by first applying $M$ physical Pauli rotations and subsequently performing syndrome extraction.
When $M$ equals the physical Pauli weight of the target logical operator, the physical rotations can be chosen to be single-qubit rotations, so the additional circuit overhead is essentially negligible.
In contrast, when $M$ is strictly smaller than the weight of the target logical operator, at least some of the physical rotations are multiqubit Pauli rotations.
Because most hardware platforms do not natively support arbitrary multiqubit Pauli rotations, a direct implementation would require decomposing them into single-qubit Pauli rotations sandwiched by entangling gates, thereby introducing additional opportunities for noise.

In this Appendix, we show how this additional overhead can be suppressed by embedding the required multiqubit rotations into the syndrome-extraction circuit.
The mechanism uses the same error-propagation structure that gives rise to hook errors; we refer to the deliberate use of this structure as \emph{hook doping}.
We illustrate the construction for the $d=7$ rotated surface code with $M=5$ in Fig.~\ref{fig:hook-doping}.
The basic idea is to push physical rotation gates inside the syndrome-extraction circuit with modified CNOT ordering so that conjugation by the surrounding CNOTs maps the rotations to the desired physical Pauli rotations.
Although the resulting hook-doped schedule reduces the memory fault distance, we find that such a  reduction does not affect the overall logical error rate in the parameter regime where the intrinsic error of the weak transversal protocol is the dominant contribution.

To implement a weak transversal gate using $M<d$ physical rotations, we intentionally introduce the corresponding hook propagation.
As an example, Fig.~\ref{fig:hook-doping}(c) shows a partial modification of the CNOT-layer ordering for $Z$-type stabilizers.
This schedule has three useful properties:
\begin{itemize}
    \item The CNOT depth is unchanged; no additional CNOT layer is required.
    \item $X$-type syndrome extraction is unaffected, so the fault distance for the logical $X$ operator remains $d=7$.
    \item Although the fault distance for the logical $Z$ operator is reduced to $5$, the effect on the overall performance is small.
\end{itemize}
For the first point, Fig.~\ref{fig:distance-7-inplace-rot-steps} shows that the schedule indeed contains no additional CNOT layers beyond the original four.
The only extra layer is an $R_{\rm Z}$ layer, which is much faster than entangling gates on most platforms and can be virtual on superconducting devices.
To verify the second and third points, we performed numerical simulations; see Fig.~\ref{fig:hook-doping-LER}.
Because $X$-type syndrome extraction is unchanged, the logical-error-rate scaling $O(p^4)$ for the $Z$-memory circuit remains consistent with fault distance $7$.
By contrast, in the partially hook-doped circuit, the fault distance for logical $Z$ errors is reduced to $5$, consistent with an $O(p^3)$ scaling.
We also simulated a fully hook-doped schedule obtained by reversing the CNOT ordering for all $Z$-type stabilizers.
In this case the fault distance drops to $\lceil d/2\rceil=4$, giving an $O(p^2)$ logical-error-rate scaling.

Although the drop of fault distance  is undesirable from the viewpoint of quantum memory, its quantitative impact is limited in the regime of interest. For instance, at $\pphys=10^{-4}$, the memory logical error rate remains below $10^{-7}$. This value is comparable to the logical error rate of in-place Pauli rotations only when the target logical angle is as small as $\bartheta\lesssim 10^{-5}$. By contrast, for practical computations using a distance-$7$ surface code, we expect the relevant regime to be $\bartheta\gtrsim 10^{-4}$, in which the intrinsic error of the weak-transversal rotation dominates over the additional memory error introduced by hook doping. Thus, in this regime, hook doping does not degrade the overall performance.

In summary, this example shows that a modified syndrome-extraction schedule can be used to implement multiqubit Pauli rotations required for weak transversal gates without adding CNOT depth.
Although the discussion here focused on a single rotated-surface-code patch, the same idea can in principle be applied to more general codes, provided that the relevant Pauli rotations can be embedded into their syndrome-extraction circuits through controlled error propagation.

\begin{figure}[t]

    \centering

    \includegraphics[width=0.98\linewidth]{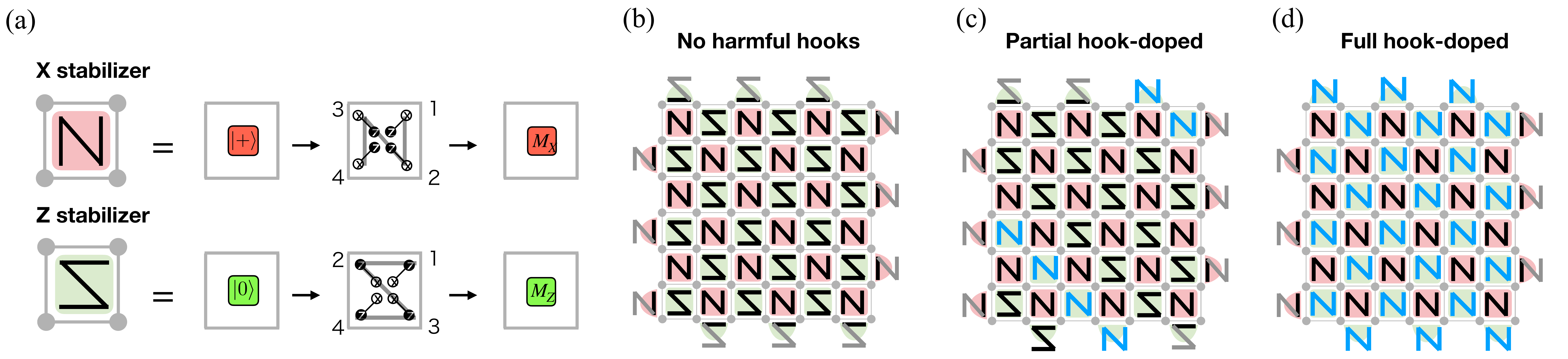}

    \caption{(a) Graphical notation for CNOT orderings for $X$- and $Z$-type stabilizer measurements. (b) An ordinary syndrome-extraction schedule that preserves the code distance of 7 at the circuit level. (c) A hook-doped schedule. The blue strokes indicate the modified parts of the schedule, which intentionally invokes harmful hooks for the sake of implementing ZZ  rotation gate. Here the fault distance is 5.
    (d) Fully hook-doped schedule. Here the fault distance is 4.
    }

    \label{fig:hook-doping}

    \vspace{1em}

    \includegraphics[width=0.98\linewidth]{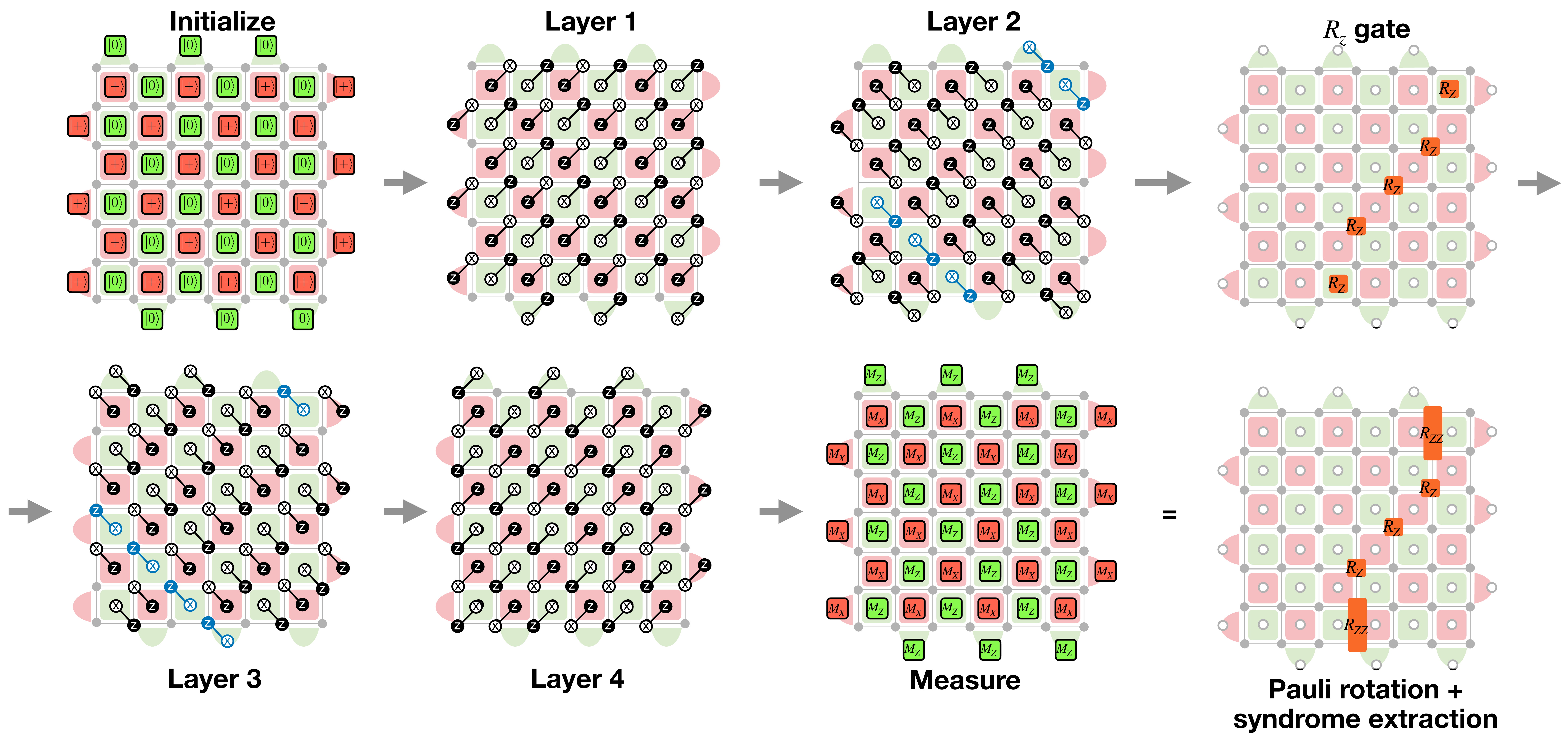}

    \caption{Full circuit description of weak transversal Z rotation for surface code of $d=7$ with $M=5$ rotation gates.}

    \label{fig:distance-7-inplace-rot-steps}

\end{figure}
\color{black}

\begin{figure}[t]
    \centering
    \includegraphics[width=0.65\linewidth]{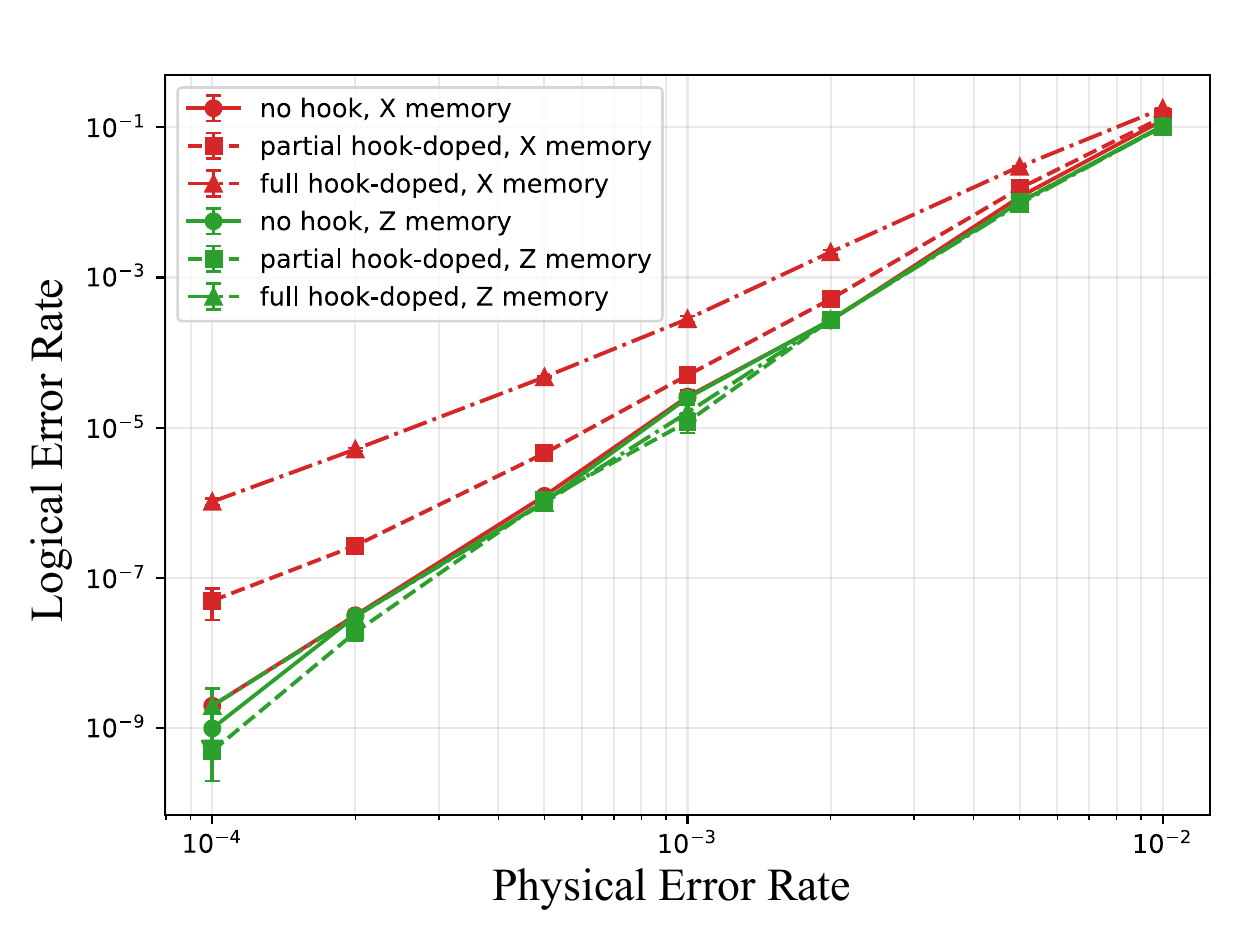}
    \caption{Logical error rate in memory experiments for the $d=7$ rotated surface code after $d$ rounds of syndrome extraction. The red and green data points show the results of $X$- and $Z$-type memory experiments, respectively. Circle, square, and triangle markers correspond to the harmful-hook-free schedule, the partially hook-doped schedule shown in Fig.~\ref{fig:hook-doping}(c), and the fully hook-doped schedule, respectively.
    }
    \label{fig:hook-doping-LER}
\end{figure}

\section{Resource estimation details for the megaquop benchmark} \label{app:resource-est-details}

In this appendix, we give the formulas used to compute the timestep $\tau$, the number of active physical qubits $N_{\rm phys}$, and the total logical error rate $\bar{p}_{\rm tot}$ for the megaquop benchmark in Sec.~\ref{subsec:resource-est}.

\subsection{Surface code}
\subsubsection{Clifford+T compilation}
\label{subsubsec:surface-clifford-t}

In general, when we consider Clifford+T compilation of a quantum circuit, all the non-Clifford gates are first synthesized so that all gates are given as a sequence of Clifford and T gates.
Regarding the benchmarked circuit, this can be done by simply performing the single-qubit gate synthesis of the rotation gate.
Here we employ the ancilla-free near-optimal synthesis algorithm proposed in Ref.~\cite{ross2014optimal}, which requires approximately $3\log_2(1/\varepsilon)+O(1)$ T gates to implement a single rotation gate up to accuracy of $\varepsilon.$
These T gates are implemented via gate teleportation consuming a single magic state that is distilled once~\cite{litinski2019game}.
For Clifford gates, we implement CNOT using lattice surgery~\cite{horsman2012surface}, Hadamard gate by transversal gate followed by logical patch rotation~\cite{fowler2018low, litinski2019game}, and S gate via lattice surgery~\cite{litinski2019game}.

That said, the total error can be estimated as follows,
\begin{align}
    \bar{p}_{\rm tot} &= \bar{p}_{\rm dec} + \bar{p}_{\rm dis} + \bar{p}_{\rm syn},
\end{align}
where $\bar{p}_{\rm dec}, \bar{p}_{\rm dis}, \bar{p}_{\rm syn}$ are decoding error, distillation error, and synthesis error, respectively~\cite{yoshioka2025error}.
Explicitly broken down, the individual expression of errors are given as follows:
\begin{align}
    \bar{p}_{\rm dec} &= N_{\rm CNOT} \cdot \bar{p}_{\rm CNOT} +N_{\rm H} \cdot \bar{p}_{\rm H} + N_{\rm S} \cdot \bar{p}_{\rm S},\\
    \bar{p}_{\rm dis} &= N_{\rm T} \cdot \bar{p}_T,\\
    \bar{p}_{\rm syn} &= N_{\rm rot} \cdot \varepsilon_{\rm syn},
\end{align}
where $N_{\rm CNOT/S/H/T/rot}$ is the number of each gate in the circuit, $\bar{p}_{\rm CNOT/S/H/T}$ is the logical error rate per gate that is estimated as given in Table~\ref{tab:gate-resource-surface-cliffordt}. The gate synthesis error per each rotation gate, $\varepsilon_{\rm syn}$, is homogeneously taken to be $10^{-4}$.

The total timesteps are estimated as follows,
\begin{eqnarray}
    \tau = \sum_{i=1}^{\# {\rm gates}} \tau_i,
\end{eqnarray}
where $\tau_i$ indicates the timestep for the $i$-th gate operation. Each $\tau$ for the instruction set architecture (ISA) is given in Table~\ref{tab:gate-resource-surface-cliffordt}.

Finally the number of active physical qubits is computed assuming the qubit plane layout as shown in Fig.~\ref{fig:architecture}(b) as
\begin{eqnarray}
    N_{\rm phys} = 1.5^2N \cdot 2d^2 + n_F A_F,
\end{eqnarray}
where $n_F$ is the number of magic state factories and $A_F$ is the number of physical qubits per magic state factory.
For the sake of comparison with the bicycle architecture (Sec.~\ref{subsubsec:gross-clifford-t}), here we take $n_F=1$, and use $A_F = 810$ following Ref.~\cite{litinski2019magic}.

\input{tables/resource_surface_cliffordt_instruction_table.tex}
\input{tables/resource_surface_cliffordphi_instruction_table}

\subsubsection{\texorpdfstring{Clifford+$\phi$}{Clifford+phi} compilation}
\label{subsubsec:surface-inplace}

When we use in-place Pauli rotation, the total error rate can be estimated as follows,
\begin{align}
    \bar{p}_{\rm tot} = \bar{p}_{\rm dec} + \bar{p}_{\rm rot},
\end{align}
where $\bar{p}_{\rm dec}$ is the decoding error as introduced in the previous section, and $\bar{p}_{\rm rot}$ is the error of in-place rotation. Using the phenomenological simulation presented in Sec.~\ref{subsec:inplace-rot}, we assume that the logical error rate per logical Z rotation is $9.6\times 10^{-5}.$
As we summarize in \cref{tab:gate-resource-surface-inplace}, the costs for Clifford gates are taken to be the same as in Clifford+T compilation.

For the timesteps, we estimate as
\begin{align}
    \tau_{\rm tot} = \sum_i \tau_{i},
\end{align}
where $\tau_i$ is the timestep required to execute the $i$-th instruction.
As summarized in \cref{tab:gate-resource-surface-inplace}, the timesteps for Clifford gates are taken to be the same as in \cref{subsubsec:surface-clifford-t}.
For the in-place Pauli rotation, we fix the number of iterations to 30 to give a conservative estimate, while in reality it may finish in fewer iterations. We foresee that dynamical recompilation of circuits~\cite{sethi2025rescq} can further reduce the runtime.
At each iteration, we assume $d$ rounds of syndrome extraction, with the first round preceded by the Pauli rotation as discussed in Sec.~\ref{app:scheduling-d7-surface}.
While such an additional operation may naively seem to increase the number of timesteps, we find that there is no overhead to perform in-place Z rotation (See Appendix~\ref{app:scheduling-d7-surface}) in this case.
Therefore, the number of timesteps is obtained by multiplying the number of code cycles by 6.

The number of active physical qubits is computed assuming the qubit plane layout shown in \cref{fig:architecture}(c) as
\begin{align}
    N_{\rm phys} = 1.5^2 N \cdot 2d^2.
\end{align}

\subsection{Gross code}
\subsubsection{Pauli-based computation}
\label{subsubsec:gross-clifford-t}

Pauli-based computation is a way of compiling a quantum circuit into one composed entirely of Pauli rotation and Pauli measurements (allowing multiqubit Paulis)~\cite{bravyi2016trading, litinski2019game}.
Pauli operators are conjugated so that Clifford gates no longer appear explicitly in the circuit. As a result, the circuit can be expressed by the instructions indicated in \cref{tab:gate-resource-gross-pbc}.

The timesteps are calculated by a public code available in Ref.~\cite{bicycle-architecture-compiler}.

The number of active physical qubits is provided as
\begin{align}
    N_{\rm phys} = N_{\rm patch} \cdot (c + u + a) - a + a' + n_F A_F,
\end{align}
where $N_{\rm patch}$ is the number of gross code modules, $c=288$ is the number of physical qubits within each gross code module, $u =90$ indicates the contribution from the logical processing unit, $a = 22$ is required for adapter between gross codes, and $a'=13$ is for adapter between a gross code and magic state factory.
We assume that the magic states are prepared by level-1 distillation for a surface code of distance 7, and hence $A_F=810$.
Since we consider the Pauli-based computation, we take $n_F=1$.

\input{tables/resource_gross_pbc_instruction_table.tex}

\subsubsection{\texorpdfstring{Clifford+$\phi$}{Clifford+phi} compilation}
\label{subsubsec:gross-inplace}

In Clifford+$\phi$ compilation, we rather aim to implement the Clifford gates using the in-module and inter-module measurements interleaved by shift automorphisms.
For instance, a CNOT gate can be decomposed into XX and ZZ measurements followed by Pauli correction.
In order to perform resource estimation of the benchmarked circuit, we first calculate the cost for performing CNOT for all neighbors as $\prod_{i:{\rm odd}} {\rm CNOT}_{i, i+1} \prod_{i:{\rm even}} {\rm CNOT}_{i, i+1}$, whose results are shown in \cref{tab:gate-resource-gross-inplace}.
For the in-place Pauli rotation, we assume the scheme uses $M=5$ physical rotations, and fix the number of iterations to 30, as in the surface-code case. However, unlike the surface code, the logical support provided in Ref.~\cite{yoder2025tour} does not allow us to implement all multi-Z rotations mediated by Z-check ancilla qubits; instead, we must use X-check ancilla qubits. This forces us to switch the basis between X and Z, and hence imposes 2 additional physical CNOT layers. Since the number of timesteps is 8 for a single round of syndrome extraction, we need 10 timesteps when we introduce physical rotation gates.

The number of active physical qubits is provided as
\begin{align}
    N_{\rm phys} = N_{\rm patch} \cdot (c + u + a) - a.
\end{align}
Note that there is no need for a magic state factory or the adapter between a gross code patch and a magic state factory.

\input{tables/resource_gross_cliffordphi_instruction_table.tex}

\section{Analysis of small angle rotation gate cost under postselection} \label{app:small-angle-gate}

Here, we explain the scaling shown in Table~\ref{tab:small-angle-subroutines} in the main text.
As mentioned, we assume here that logical qubits are encoded into surface code of distance $d$.
Therefore, the fallback recursion of teleportation-based rotation gates terminates at the fault-tolerant $S$ gate.

With angle-agnostic injection, an arbitrary single-qubit state can be injected into the surface code with a logical error rate of $\frac{2}{5}p$ and logical qubitcycle cost of $O(d)$~\cite{li2015magic}.
The injected states may also be distilled using the Meier-Eastin-Knill (MEK) protocol~\cite{meier2012magic, campbell2016_efficient}, which uses the $[[4,2,2]]$ code to obtain quadratic error suppression.
Writing the cost of an $R_l$ gate as $\mathfrak{C}_{R_l}$ and the nonrecursive MEK overhead as $\mathfrak{C}_{{\rm MEK}_l}$, the fallback recursion gives
\begin{align} \label{eq:mek-cost}
\mathfrak{C}_{R_l}
    = \mathfrak{C}_{|M_l\rangle} + 0.5 \mathfrak{C}_{R_{l-1}}
    = \mathfrak{C}_{R_{l-1}} + 0.5\mathfrak{C}_{{\rm MEK}_l}.
\end{align}
For angle-agnostic injection, $\mathfrak{C}_{{\rm MEK}_l}=O(d)$ is independent of $l$, giving $\mathfrak{C}_{R_l}=O(dl)$.

Angle-dependent injection improves the error of the injected magic state to $O(\pphys\bartheta^\alpha)$ with $\alpha=2(1-1/M)$.
However, after the fallback recursion, the logical error of the implemented gate scales only as $O(\pphys\bartheta)$~\cite{toshio2025practical}, unless one relies on unitary synthesis~\cite{toshio2026star}.
The lower-level fallback rotations dominate because the fallback probability is reduced by a factor of two per level, whereas the angle-dependent error grows by a factor $2^\alpha>2$ when moving down one level.
The numerical data in Fig.~\ref{fig:state-prep}(a) confirm this behavior, and the circuit volume approaches
\begin{eqnarray}
    14d + O(d^2 2^{-2l/d}),
\end{eqnarray}
which is the sum of state preparation with cost $V_L\sim d+O(d^2 2^{-2l/d})$, a CNOT gate with $V_L\sim 4d$, and measurement feedback with cost $V_L\sim 2d$. Since the expected number of fallback steps is two, we obtain the above expression.

For the weak transversal gates, the relevant acceptance event is the absence of a syndrome.
For an $M$-rotation weak transversal gadget realizing a logical angle $\bartheta$, the no-syndrome probability is
\begin{align}
    P_{\rm success} = p_{\chi = 0} = |\cos \theta|^{2M} + |\sin \theta|^{2M} \approx 1 - M\bar{\theta}^{2/M}.
\end{align}
The dominant contribution to the logical error is the $\chi=1$ event being misclassified as $\chi=0$. This occurs with probability of $p_{\chi=1}\cdot (M \pphys)=O(\pphys M \theta^2)$, and results in diamond norm error of $O(\theta^{2M-4})$, and hence the overall error is $O(\pphys M\theta^{2M-2}) = O(\pphys M \bartheta^{2(1-1/M)})$.
Due to the superlinear scaling with the angle, the logical error can be reduced by segmenting the desired rotation into $r$ rotations of angle $\bartheta/r$:
\begin{align}
\left({\rm Diamond~norm~error}\right) = O\left(r\cdot \pphys M \left(\bartheta/ r\right)^{2(1-1/M)}\right) = O(\pphys M\bartheta^{2(1-1/M)} r^{2/M-1}).
\end{align}
The corresponding success probability is
\begin{align}
P_{\rm success} \approx \left(1 - M \left(\frac{\bar{\theta}}{r}\right)^{2/M}\right)^r,
\label{eq:psuccess_weak}
\end{align}
so the expected spacetime cost is estimated as
\begin{align}
    \mathfrak{C}_{R_l} = rd/P_{\rm success} = rd + O(dMr^{2-2/M}2^{-2l/M}),
\end{align}
for $l\gg 1$. Note that this cost analysis assumes that the spacetime cost of the gate preceding the weak transversal gate is negligibly small, such as graph state preparation.
The segmentation parameter therefore controls the error-volume tradeoff: larger $r$ improves the logical error but also increases the baseline volume $rd$.

\section{Meier-Eastin-Knill protocol for distillation of equatorial magic states} \label{app:MEK}

Here, we provide a brief review on the distillation protocol for equatorial states proposed by Meier-Eastin-Knill~\cite{meier2012magic} and later improved by Campbell and O'Gorman~\cite{campbell2016_efficient}. We hereby refer to the protocol as the MEK protocol. We denote by $|M_k\rangle = R_k|+\rangle = R_{\rm Z}\left(\frac{\pi}{2^k}\right) |+\rangle$ the magic state for Pauli rotation gate from the $k$-th level of Clifford hierarchy, and also denote by $H_k = R_k X R_k^\dagger$ the unitary and Hermitian operator that stabilizes $|M_k\rangle$.
In short, the MEK protocol performs a parity check on $|M_k\rangle$ regarding $H_l$, using the four qubit code which is also known as the $C_4$ code (see Fig.~\ref{fig:MEK}).

Thanks to the [[4, 2, 2]] code, the errors of non-Clifford $R_3$ gates can be suppressed up to quadratic factor. On the other hand, the non-Clifford $R_{l-1}^\dagger$ is not encoded, and hence contributes linearly to the total error rate.
By performing numerical analysis neglecting the correlations of noise, we find that the total error rate per distilled magic state $\epsilon'_l$ can be written as
\begin{align}
    \epsilon'_l =
    8 \epsilon_3^2 + \epsilon_l^2 +  \epsilon^2_2
    + \frac{1}{4} \eta_{l-1} + \frac{3}{20} \eta_F + \frac{1}{16}\eta_{\rm CNOT},
\end{align}
where $\epsilon_l$ and $\epsilon_{\rm CNOT}$ are the error rate per magic state for the $l$-th level of Clifford hierarchy and CNOT gate, $\eta_{l-1}$ is the error rate of the $R_{l-1}$ gate, $\eta_{F}$ is the error rate of the encoding circuit into the four qubit code, and $\eta_{\rm CNOT}$ is the error rate of the CNOT layer.
Note that $\epsilon$ and $\eta$ indicate whether the error is quadratically suppressed or not.
For instance, the input magic state infidelity and T gate error can be suppressed quadratically, while the error in the mid-circuit $R_{l-1}$ cannot be.
Due to the fact that the $R_{l-1}$ gate is not protected by the [[4,2,2]] code, we must ensure that the gate is implemented with high accuracy, and hence it is called the ``pivot gate.''
This has motivated us to perform an additional round of distillation for the magic state corresponding to $R_{l-1}$.

\begin{figure*}[h]
    \centering
\includegraphics[width=0.9\linewidth]    {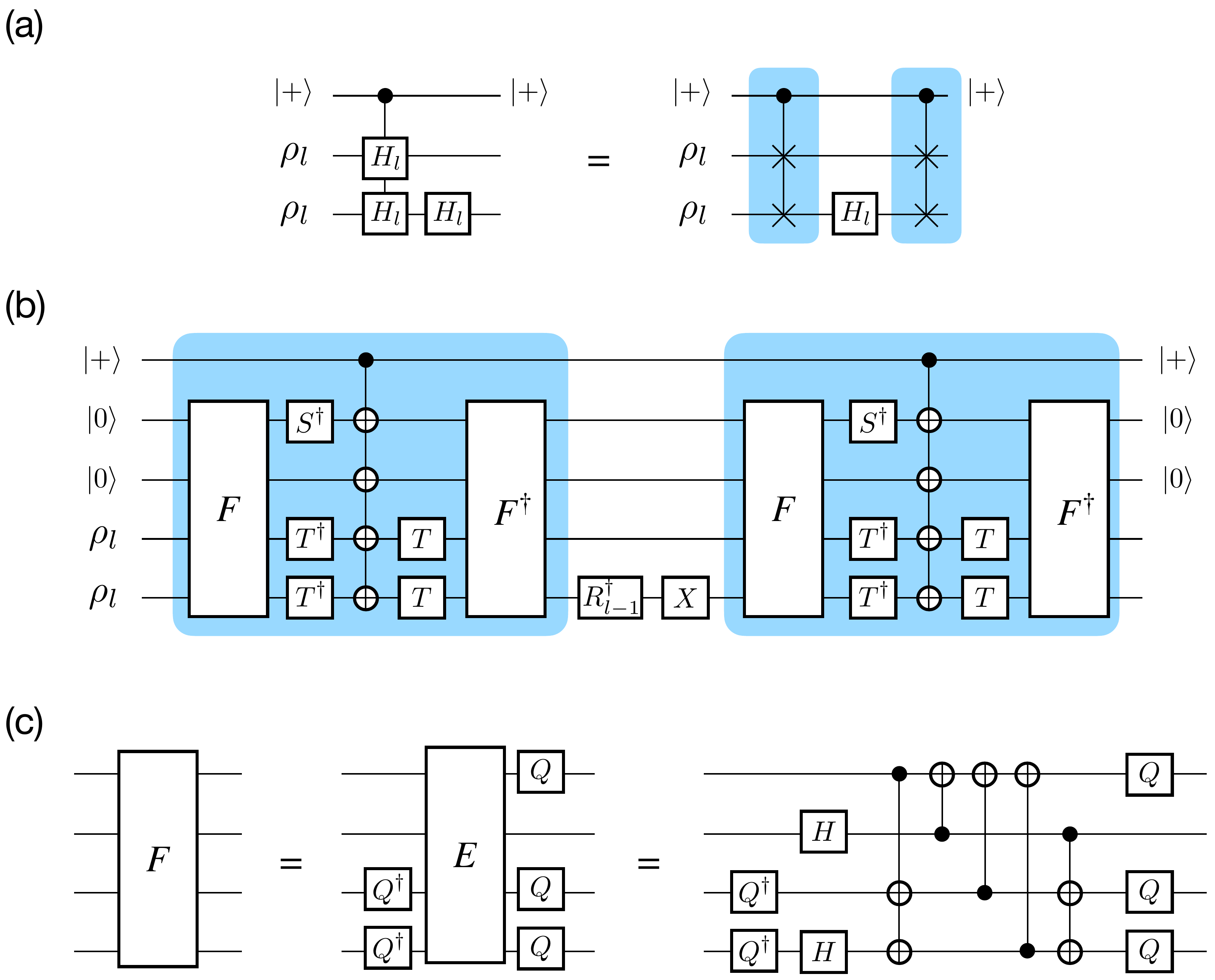}
    \caption{
(a) Hadamard test regarding a unitary and Hermitian operator $H_l= R_l X R_l^\dagger$, which stabilizes $|M_l\rangle$, over two noisy magic states $\rho_l$. (b) The Hadamard test with code concatenation into [[4,2,2]] code, in which the transversal implementation of Hadamard gate realizes a logical SWAP gate. Note that some rotation gates can be canceled out~\cite{campbell2016_efficient}, which results in consuming eight T gates and one $R_{l-1}$ gate. The blue region corresponds to the controlled SWAP gate as indicated in panel (a). The operator is rewritten using the identity $H_l = R_l X R_l^\dagger = X R_{l-1}^\dagger$. (c) Detailed circuit structure of the encoding circuit $F$. Here, we have introduced a unitary $Q$ that conjugates Pauli $Y$ and $Z$ as $Q Y Q^\dagger = Z$ and hence also $Q R_{\rm Y}(\theta) Q^\dagger = R_{\rm Z} (\theta)$.}
    \label{fig:MEK}
\end{figure*}

\bibliography{bib}

\end{document}

%% file: tables/resource_estimation_total_table.tex
\begin{table*}[t]
\centering
\begin{tabular}{|c|c|c|c|c|c|c|}
\hline
Architecture & Code & \# module & \# factory & Timestep $\tau$ & Active physical qubits $N_{\rm phys}$ & Total error rate $\bar{p}_{\rm tot}$ \\
\hline \hline
Clifford+T & \multirow{7}{*}{Surface} & 108 & 1 & $124{,}216{,}320$ & $24{,}624$ & $1.12191$ \\
\cline{1-1}\cline{3-7}
Clifford+T &  & 108 & 2 & $62{,}124{,}960$ & $25{,}434$ & $1.12191$ \\
\cline{1-1}\cline{3-7}
Clifford+T &  & 108 & 4 & $31{,}079{,}280$ & $27{,}054$ & $1.12191$ \\
\cline{1-1}\cline{3-7}
Clifford+T &  & 108 & 8 & $15{,}556{,}440$ & $30{,}294$ & $1.12191$ \\
\cline{1-1}\cline{3-7}
Clifford+T &  & 108 & 16 & $7{,}795{,}020$ & $36{,}774$ & $1.12191$ \\
\cline{1-1}\cline{3-7}
Clifford+T &  & 108 & 32 & $3{,}914{,}310$ & $49{,}734$ & $1.12191$ \\
\cline{1-1}\cline{3-7}
Clifford+T &  & 108 & 64 & $1{,}973{,}955$ & $75{,}654$ & $1.12191$ \\
\hline
Clifford+T & Gross & 10 & 1 & $194{,}998{,}776$ & $4{,}801$ & $0.83833$ \\
\hline \hline
Clifford+$\phi$ & Surface & 108 & 0 & $159{,}600$ & $23{,}814$ & $1.03742$ \\
\hline
Clifford+$\phi$ & \multirow{4}{*}{Gross} & 10 & 0 & $4{,}426{,}800$ & {\bf $3{,}991$} & $1.03714$ \\
\cline{1-1}\cline{3-7}
Clifford+$\phi$ & & 20 & 0 & $1{,}531{,}200$ & $7{,}991$ & $1.03707$ \\
\cline{1-1}\cline{3-7}
Clifford+$\phi$ & & 30 & 0 & $1{,}070{,}400$ & $11{,}991$ & $1.03713$ \\
\cline{1-1}\cline{3-7}
Clifford+$\phi$ & & 40 & 0 & $746{,}400$ & $15{,}991$ & $1.03713$ \\
\hline
\end{tabular}
\caption{
Total resource estimates for the benchmark circuit of Fig.~\ref{fig:architecture}(a), with $N=108$ logical qubits, $100$ circuit layers, $\bar{\theta}=0.001$, and $\pphys=10^{-4}$.
For surface-code Clifford+T compilation, the table sweeps $1,2,4,8,16,32,64$ distance-$7$ magic-state factories at synthesis precision $\varepsilon_{\rm syn}=10^{-4}$, assuming that the rotation-synthesis contribution to the timestep is parallelized by the number of factories while the first-order logical error is set by the same gate count.
For gross-code Clifford+T compilation, we only report the single-factory estimate, since increasing the number of factories does not directly translate into the same simple parallelization model. The gross-code Clifford+T row uses $\lceil N/11\rceil=10$ gross modules plus one factory, because one logical slot per module is reserved as a pivot/mediator slot in this comparison. The 10-module gross-code Clifford+$\phi$ row uses the same number of gross modules without a factory.
For gross-code Clifford+$\phi$ compilation, the listed variants use the locked-pivot estimates for the indicated number of gross modules.
The total error rate is the first-order sum $\bar{p}_{\rm tot}=\sum_i \bar{p}_i$.
}
\label{tab:resource-est-total}
\end{table*}

%% file: tables/accuracy-small-angle.tex
\begin{table}[b]
    \centering
\begin{tabular}{c|c|c}
Method                                                          & Logical Error Rate                 & Logical qubitcycle $V_L$  \\
\hline \hline
Teleport, Li~\cite{li2015magic}                                                            & $O(p)$              & $O(d)$  \\
\hline
Teleport, Li+MEK~\cite{meier2012magic, campbell2016_efficient}                         & $O(p^2)$            & $O(dl)$ \\
\hline
Teleport, transversal injection~\cite{choi2023fault, toshio2025practical}                                         & $O(p 2^{-l})$       & $14 d +O(d^2 2^{-2l/d})$   \\
\hline \hline
\begin{tabular}[c]{@{}c@{}}Weak transversal gate\\ (general, with postselection)\end{tabular}    & $O(p2^{-2l(1-1/M)}r^{2/M-1})$ & $rd + O(dMr^{2-2/M}2^{-2l/M})$    \\
\hline
\end{tabular}
    \caption{Asymptotic scaling of logical error rate and logical qubitcycle for small-angle gadget of $l(\gg 1)$-th level of Clifford hierarchy. We assume that logical qubit is of distance $d$.
    The weak transversal row assumes $r$ segments and postselection.
    }
    \label{tab:small-angle-subroutines}
\end{table}

%% file: tables/surface_m3_implemented_protocol_table.tex
\begin{table*}[t]
\centering
\begin{tabular}{|c|c|c|c|c|}
\hline
Target $\bartheta$ & Logical error $\bar p_Z^{\rm eff}$ & Decoder & Best $R$ & Policy \\
\hline \hline
$5\times 10^{-5}$ & $(3.00\pm1.73)\times 10^{-6}$ & 1-step & 1 & $\chi_{\rm target}=[0]$; optimized \\
\hline
$1\times 10^{-4}$ & $(3.99\pm2.00)\times 10^{-6}$ & 1-step & 1 & $\chi_{\rm target}=[0]$; optimized \\
\hline
$2\times 10^{-4}$ & $(1.50\pm0.39)\times 10^{-5}$ & 1-step & 1 & $\chi_{\rm target}=[0]$; optimized \\
\hline
$3\times 10^{-4}$ & $(2.39\pm0.49)\times 10^{-5}$ & 1-step & 1 & $\chi_{\rm target}=[0]$; optimized\\
\hline
$5\times 10^{-4}$ & $(5.48\pm0.74)\times 10^{-5}$ & 1-step & 1 & $\chi_{\rm target}=[0]$; optimized \\
\hline
$1\times 10^{-3}$ & $(1.30\pm0.11)\times 10^{-4}$ & 1-step & 1 & $\chi_{\rm target}=[0]$; optimized \\
\hline
$2\times 10^{-3}$ & $(3.52\pm0.19)\times 10^{-4}$ & 2-step & 15 & $\chi_{\rm target}=[0,0,1]$; unoptimized \\
\hline
$5\times 10^{-3}$ & $(4.48\pm0.22)\times 10^{-4}$ & 2-step & 28 & $\chi_{\rm target}=[0,1]$; unoptimized \\
\hline
$1\times 10^{-2}$ & $(5.06\pm0.25)\times 10^{-4}$ & 2-step & 21 & $\chi_{\rm target}=[0,0,1]$;unoptimized \\
\hline
$1.5\times 10^{-2}$ & $(6.37\pm0.29)\times 10^{-4}$ & 2-step & 15 & $\chi_{\rm target}=[0,1]$; unoptimized \\
\hline
$2\times 10^{-2}$ & $(5.65\pm0.31)\times 10^{-4}$ & 2-step & 20 & $\chi_{\rm target}=[0,1]$; unoptimized \\
\hline
$3\times 10^{-2}$ & $(7.22\pm0.40)\times 10^{-4}$ & 2-step & 20 & $\chi_{\rm target}=[0,1]$; unoptimized \\
\hline
\end{tabular}
\caption{
Implemented protocols used for the rotated surface-code circuit-level simulations with $d=M=3$.
The logical error is the $Z$-based effective logical error from the hook-aware circuit level simulation for surface code using MPS at $\pphys=10^{-4}$ with depolarizing physical noise; the uncertainty denotes one standard error from the $10^6$-shot estimates.
$\chi_{\rm target}$ indicates the target value of $\chi$ used at $r$-th round. 
The ``optimized" policy indicates that physical angles are optimized to minimize the diamond distance for circuit-level density matrix simulation using $d=3$ repetition code.
}
\label{tab:surface-m3-implemented-protocol}
\end{table*}

%% file: tables/resource_surface_cliffordt_instruction_table.tex
\begin{table}[t]
\begin{tabular}{|c|c|c|c|}
\hline
Instruction         & Cycles $t_i$ & Timestep $\tau_i$ & Logical Error $\bar{p}_i$       \\
\hline \hline
CNOT              & 14   & 84    & $1.4\times 10^{-8}$  \\
\hline
Hadamard              & 21   & 126    & $2.1\times 10^{-8}$  \\
\hline
S              & 7   & 42    & $7.0\times 10^{-9}$  \\
\hline
T              & 18.1   & 109    & $6.2\times 10^{-8}$  \\
\hline
\end{tabular}
\caption{Cost for elementary instructions in Clifford+T compilation using surface code of code distance $d=7$ on 2d planar architecture. Here, logical error rates for Clifford gates are computed as $\bar{p} \cdot t_i$ where $\bar{p}=0.1(\pphys/p_{\rm th})^{\frac{d+1}{2}}$ is the logical error rate per code cycle under threshold error rate of $p_{\rm th}=0.01$.
For physical error rate of $\pphys = 10^{-4}$ and code distance of $d=7$, the logical error rate per cycle is $ \bar{p}= 1.0\times 10^{-9}$.
The number of code cycles, $t_{\rm CNOT/H/S}$ are $2d$, $3d$, and $d$ for CNOT, H, and S gates, respectively~\cite{litinski2019game}.} \label{tab:gate-resource-surface-cliffordt}
\end{table}

%% file: tables/resource_surface_cliffordphi_instruction_table.tex
\begin{table}[t]
\begin{tabular}{|c|c|c|c|}
\hline
Instruction         & Cycles $t_i$ & Timestep $\tau_i$ & Logical Error $\bar{p}_i$       \\
\hline \hline
CNOT              & 14   & 84    & $1.4\times 10^{-8}$  \\
\hline
\begin{tabular}[c]{@{}c@{}}In-place rotation\\ ($\theta=0.001$)\end{tabular} & -  & 180   & $9.6\times 10^{-5}$\\
\hline
\end{tabular}
\caption{Cost for elementary instructions of Clifford+$\phi$ compilation in surface code of code distance $d=7$ on 2d planar architecture.}
\label{tab:gate-resource-surface-inplace}
\end{table}

%% file: tables/resource_gross_pbc_instruction_table.tex
\begin{table}[t]
\begin{tabular}{|c|c|c|}
\hline
Instruction        & Timestep $\tau_i$ & Logical Error $\bar{p}_i$       \\
\hline \hline
idle                & 8    & $1.4\times 10^{-15}$  \\
\hline
shift automorphism                & 12    & $6.1\times 10^{-14}$  \\
\hline
in-module meas.                & 120    & $1.0\times 10^{-9}$  \\
\hline
inter-module meas.                & 120    & $4.8\times 10^{-8}$  \\
\hline
T injection                & 73+120    & $8.8\times 10^{-7}$  \\
\hline
\end{tabular}
\caption{Cost for elementary instructions of Pauli-based computation in gross code~\cite{yoder2025tour}. }
\label{tab:gate-resource-gross-pbc}
\end{table}

%% file: tables/resource_gross_cliffordphi_instruction_table.tex
\begin{table}[t]
\centering
\begin{tabular}{|c|c|c|c|}
\hline
Instruction & Max. logical per patch & Timestep $\tau_i$ & Logical Error $\bar{p}_i$ \\
\hline \hline
idle & Any & 8 & $1.44\times 10^{-15}$ \\
\hline
\multirow{4}{*}{CNOT for all neighbors}
  & 11 & 42768 & $3.4\times 10^{-6}$ \\
\cline{2-4}
 & 6 & 17976 & $3.2\times 10^{-6}$ \\
\cline{2-4}
 & 4 & 10200 & $3.0\times 10^{-6}$ \\
\cline{2-4}
 & 3 & 6864 & $3.3\times 10^{-6}$ \\
\hline
\multirow{4}{*}{Simultaneous in-place $Z$ rot.}
  & 11 & 1500 & $1.0\times 10^{-2}$ \\
\cline{2-4}
 & 6 & 600 & $1.0\times 10^{-2}$ \\
\cline{2-4}
 & 4 & 600 & $1.0\times 10^{-2}$ \\
\cline{2-4}
 & 3 & 600 & $1.0\times 10^{-2}$ \\
\hline
\end{tabular}
\caption{Cost of instructions in the Clifford+$\phi$ formalism for the gross code. The logical rotation angles are homogeneously taken as $\bar{\theta}=0.001$.}
\label{tab:gate-resource-gross-inplace}
\end{table}